\documentclass[onecolumn,showpacs,preprint,amsmath,amssymb,superscriptaddress]{revtex4-1}

\usepackage{dcolumn}
\usepackage{bm}

\newcommand{\ignore}[1]{}
\usepackage{graphicx}
\usepackage{amsmath,amssymb,amsfonts,amsthm,color,latexsym}
\usepackage{soul}
\usepackage[caption=false]{subfig}
\def\mb{\mathbf}

\def\bm{\boldsymbol}
\newcommand{\Rmnum}[1]{\uppercase\expandafter{\romannumeral #1\relax}}
\def\mb{\mathbf}
\usepackage{amsmath}
\newcommand*\diff{\mathop{}\!\mathrm{d}}

\mathchardef\mhyphen="2D
\usepackage{amsmath}
\DeclareMathOperator{\Tr}{Tr}
\newtheorem{theorem}{Theorem}[section]

\newtheorem{proposition}[theorem]{Proposition}

\usepackage[normalem]{ulem}

\begin{document}
\preprint{}

\title{Machine learning based non-Newtonian fluid model with molecular fidelity}
\author{Huan Lei}
\email{leihuan@msu.edu}
\affiliation{Department of Computational Mathematics, Science \& Engineering and Department of Statistics \& Probability, Michigan State University, MI 48824, USA}%
\author{Lei Wu}
\affiliation{Department of Mathematics and Program in Applied and Computational Mathematics, \\Princeton University, 
NJ 08544, USA}%
\author{Weinan E}
\email{weinan@math.princeton.edu}
\affiliation{Department of Mathematics and Program in Applied and Computational Mathematics, \\Princeton University, 
NJ 08544, USA}%


\begin{abstract}
We introduce a machine-learning-based framework 
for constructing continuum non-Newtonian fluid dynamics model directly from a micro-scale 
description. Dumbbell polymer solutions are used as examples to demonstrate the essential ideas.
To faithfully retain molecular fidelity, we establish a micro-macro correspondence
via a set of encoders for the micro-scale polymer configurations
and their macro-scale counterparts, a set of nonlinear conformation tensors.
The dynamics of these conformation tensors can be derived 
from the micro-scale model and the relevant terms can 
be parametrized using machine learning.
The final model, named the deep non-Newtonian model (DeePN$^2$), takes
the form of conventional non-Newtonian fluid dynamics models, with a new form of the
objective tensor derivative.
Both  the formulation of the dynamic equation and the neural network representation rigorously preserve the rotational
invariance, which ensures the admissibility of the constructed model.
Numerical results demonstrate the accuracy of DeePN$^2$,
where models based on empirical closures show limitations.

\end{abstract}


\maketitle
\section{Introduction}
Accurate modeling of non-Newtonian fluid flows has been a long-standing problem. 
Existing hydrodynamic models have to resort to ad hoc assumptions either directly at the macro-scale level when writing down
constitutive laws, or as closure assumptions when deriving macro-scale models from some underlying micro-scale description.
A  variety of  empirical 
constitutive models \cite{Larson88,Owens_Phillips_2002} of both integral and derivative types have been developed, including  Oldroyd-B 
\cite{Oldroyd_Wilson_PRSLA_1950}, Giesekus \cite{Giesekus_JNNFM_1982}, finite 
extensible nonlinear elastic Peterlin (FENE-P) \cite{Peterlin_Polymer_Science_1966,
Bird_Doston_JNNFM_1980}, Rivlin-Sawyers \cite{Riv_Sawy_AnnFluid_1971}.
These models are designed such that 
proper frame-indifference is satisfied, but otherwise left with few physical constraints.
Despite their broad applications, the robustness and universal applicability
of these models are still in doubt.
In principle, viscoelastic effects are determined
by the polymer configuration distribution, which can be obtained by directly solving the micro-scale Fokker-Planck
equation coupled with the macro-scale hydrodynamic equation \cite{Fan_Acta_1989}.  
However, the cost of such an approach becomes prohibitive for large scale simulations due to the high-dimensionality
of the FK equation. 
Semi-analytical closures \cite{Warner_IECF_1972, Warner_PhD_1971, FENE_L_S_JNNFM_1999, 
Yu_Du_mms_2005, Hyon_Du_mms_2008} based on moment approximations of the configuration 
distribution were developed for dumbbell systems. Applications to non-steady flows \cite{Warner_IECF_1972, Warner_PhD_1971} and 
more complex intramolecular potential \cite{FENE_L_S_JNNFM_1999, Yu_Du_mms_2005, Hyon_Du_mms_2008} remains largely open due to the 
high-dimensionality of the configuration space.
Several alternative approaches \cite{Laso_Ottinger_JNNFM_1993, Hulsen_Heel_JNNFM_1997, REN_E_HMM_complex_fluid_2005} 
based on sophisticated coupling between the micro- and macro-scale  models have been proposed. 
However, the efficiency and accuracy of these approaches rely on a separation between the relevant macro- and
micro-scales, something that does not usually happen in practice.

Motivated by the recent successes in applying machine learning (ML) to construct  reduced dynamics of 
complex systems \cite{ma2018model, Vlachas_Byeon_PRSA_2018,
Han_Ma_PNAS_2019, ling_kurzawski_JFM_2016, Wang_Wu_Xiao_PRF_2017, Lusch_Kutz_Brunton_Nature_2018,
Linot_Graham_2019, Raissi_Kar_2020}, we aim to learn
accurate and admissible non-Newtonian hydrodynamic models directly from a micro-scale 
description.
However, we note that directly applying machine learning to construct such first-principled based fluid models is highly
non-trivial. The major challenge lies in how to formulate the micro-macro correspondence in a natural way, 
such that the constructed macroscale model can faithfully retain the viscoelastic properties on the 
molecular-level fidelity. Moreover, the 
deep neural network (DNN) representations need to rigorously preserve the physical symmetries. 
Secondly, to construct the 
governed reduced dynamics, most of the current ML-based approaches rely on the time-series samples and the 
various delicate numerical treatments to evaluate the time derivatives (e.g., see discussion \cite{Rudy_Kutz_Science_Ad_2017}). 
However, the microscale simulation data of non-Newtonian fluids are often limited by the affordable computational resource
and superimposed with noise (e.g., due to thermal fluctuations). It is generally impractical to obtain the accurate 
macro-scale time derivative information from the training data. Moreover, the objective tensor derivative
in existing models is chosen empirically, e.g., upper-convected \cite{Oldroyd_Wilson_PRSLA_1950}, 
covariant \cite{Oldroyd_Wilson_PRSLA_1950}, corrotational \cite{Zaremba_1903}, to ensure
the rotational symmetry constraint. Such ambiguities will be inherited if we directly learn the dynamics from the
time-series samples. A third challenge is the model interpretability, a well-known  weakness of 
machine-learning-based models. 

In this study, we present a machine-learning-based approach, the deep
non-Newtonian model (DeePN$^2$), for learning the non-Newtonian hydrodynamic model directly from the
microscale description. 
To address the aforementioned challenges, in DeePN$^2$, we learn a set of encoder functions directly
from  micro-scale simulation data, which can be used to extract the ``features'' of sub-grid 
polymer configuration. Such features are essentially the macro-scale conformation tensors which are used 
in the construction of the constitutive laws. To retain the molecular level fidelity, the second idea is 
to formulate the ansatz of reduced dynamics  directly from the micro-scale Fokker-Planck equation. 
The learning with the ansatz only requires the microscale configuration samples without the need of the time-series 
training data. Thirdly, to ensure the model admissibility, we propose a general  
symmetry-preserving DNN structure to represent the terms in the reduced dynamics. 

All these are done in an end-to-end fashion, by simultaneously learning the micro-scale encoders, 
the polymer stress, and the evolution dynamics of the macro-scale conformation tensors. 
The constructed model takes a form similar to the traditional hydrodynamic model and retains clear 
physical interpretations for individual terms. The conformation tensors are a natural extension of 
the end-end orientation tensor used in classical rheological models. 
A new objective tensor derivative naturally arises in this way. It takes a different form from the 
current choices in those empirical macroscale models. It has a unique micro-scale interpretation and 
can be systematically constructed without ambiguity. Numerical results demonstrate the accuracy 
of this machine-learning-based model as well as the crucial role of the constructed tensor derivatives 
encoded with the molecular structure.

\section{Machine-learning based non-Newtonian hydrodynamic model}
\subsection{Generalized hydrodynamic model}
Let us start with the continuum level description of the dynamics of incompressible non-Newtonian flow in the following
generalized form
\begin{equation}
\begin{split}
\nabla \cdot \mb u &= 0 \\
\rho \frac{\rm d \mb u}{\rm d t} &= -\nabla p + \nabla \cdot (\bm\tau_{s} + \bm\tau_{\rm p}) + \mb f_{\rm ext},
\end{split}
\label{eq:momentum_transport_close}
\end{equation}
where $\rho$, $\mb u$ and $p$ represent the fluid density, velocity and pressure field, respectively. $\mb f_{\rm ext}$ is the external body force and 
$\bm\tau_s$ is the solvent stress tensor with shear viscosity $\eta_s$, which is assumed to take the Newtonian form 
$\bm \tau_s  = \eta_s(\nabla \mb u + \nabla \mb u^T)$. $\bm\tau_{\rm p}$ is the polymer stress tensor whose constitutive law
is generally unknown. 
To close Eq. \eqref{eq:momentum_transport_close}, traditional models,  e.g., Oldroyd-B, Giesekus, and FENE-P, are 
generally based on the approximation of $\bm \tau_{\rm p}$ in terms of an empirically chosen conformation tensor (e.g., the end-end orientation 
tensor), along with some heuristic closure assumption 
for the dynamics of such a tensor.  


To map the  microscopic model to the continuum model \eqref{eq:momentum_transport_close}, we assume that (\Rmnum{1}) 
the polymer solution can be treated as nearly incompressible 
on the continuum scale; and (\Rmnum{2}) the polymer solution is semi-dilute, i.e., the polymer stress $\bm \tau_{p}$  is dominated by 
intramolecular interaction $V_{\rm b}(r)$, where $r = \vert \mb r\vert$ and $\mb r$ is the end-end vector 
between the two beads of a dumbbell  molecule. The form of $V_{\rm b}(r)$ will be specified later. 
The current approach can be applied to more complicated systems, see  
Appendix for results of three-bead suspension model.
Theoretically, the instantaneous $\bm \tau_{p}$  can be determined by the probability 
density function $\rho(\mb r, t)$. In DeePN$^2$, instead of directly constructing $\rho(\mb r,t)$,
we seek a micro-macro correspondence 
that maps the polymer configurations to a set of conformation tensors, by which we 
construct the stress model and the evolution dynamics, i.e., 
\begin{subequations}
\begin{align} 
&\bm \tau_p = \mb G(\mb c_1, \mb c_2, \cdots \mb c_n) \label{eq:stress_model} \\
&\frac{\mathcal{D}{\mb c}_i}{\mathcal{D} t} = \mb H_i(\nabla \mb u, \mb c_1, \cdots, \mb c_n),\label{eq:c_evolution} 
\end{align}
\label{eq:moment_model}
\end{subequations}
where $\mb c_i \in \mathbb{R}^{3\times3}$ represents the $i\mhyphen$th conformation tensor of the polymer 
configurations within the local volume unit. 
It represents the macroscale features by which we construct 
the polymer stress $\bm \tau_{\rm p}$ and the evolution dynamics. The detailed formulation will be specified
later. 
In particular, if we choose $n=1$, $\mb c_1$ the orientation tensor and approximate $\mb G(\mb c_1)$ with 
the linear or the mean-field approximation, \eqref{eq:stress_model} recovers the empirical Hookean and FENE-P model. 
Moreover, we emphasize that $\left\{\mb c_i\right\}_{i=1}^{n}$ are \emph{not the standard high-order moments for 
the closure approximations of the microscale configuration density} $\rho(\mb r, t)$
(e.g., see \cite{Warner_IECF_1972, Warner_PhD_1971, FENE_L_S_JNNFM_1999, Yu_Du_mms_2005, Hyon_Du_mms_2008}). 
They are the nonlinear conformation tensors directly learnt from the microscale samples 
for the approximation of stress $\bm\tau_{\rm p}$, rather than the recovery of the high-dimensional distribution $\rho(\mb r, t)$.

In principle, 
with certain pre-assumptions on the formulation of $\mb c_i$, a straightforward approach is to learn \eqref{eq:moment_model} 
on the macro-scale level as a ``black-box'' using time-series samples from microscale simulations. However, this requires 
the explicit form of the objective tensor derivative $\frac{\mathcal{D}{\mb c}_i}{\mathcal{D} t}$, as well as
the accurate evaluation of time-derivatives from the time-series samples. Unfortunately, both conditions become impractical
for the micro-scale non-Newtonian fluid simulations 
Alternatively, we employ machine learning to establish a micro-macro
correspondence and derive the ansatz of \eqref{eq:moment_model} directly from the micro-scale descriptions.

\subsection{Modeling ansatz derived from the micro-scale description}
To faithfully retain the micro-scale molecular fidelity, we construct $\left\{\mb c_i\right\}_{i=1}^n$ by
directly learning a set of encoders from microscale samples, i.e.,
\begin{equation}
\begin{split} 
\mb c_i = \langle \mb B_i(\mb r) \rangle \quad \mb B_i = \mb f_i (\mb r) \mb f_i^T(\mb r) \quad \quad i = 1, 2, \cdots, n, 
\end{split}
\label{eq:encoder_model}
\end{equation}
where $\mb B_i$ is a microscale encoder function that maps the microscale polymer configuration to the macroscale
feature $\mb c_i$. It has an explicit micro-scale interpretation --- the average of the $i\mhyphen$th 
second-order tensor $\mb B_i$ with respect to the encoder vector $\mb f_i(\mb r): \mathbb{R}^3 \to \mathbb{R}^{3}$.

One reason for the choice of $\mb B_i(\mb r)$ as a second-order tensor is as follows. 
The stress model $\mb G(\cdot)$ needs to retain the rotation symmetry. 
As the input of the stress model $\mb G(\cdot)$, $\mb B_i(\mb r)$ 
needs to retain  rotational symmetry in accordance
with the polymer configuration $\mb r$. For example, the
vector form of $\mb B_i(\mb r)$ needs to satisfy $\mb B_i(\mb Q\mb r) = \mb Q\mb B_i(\mb r)$ for any unitary 
matrix $\mb Q$. This yields $\left\langle \mb B_i(\mb r)\right\rangle \equiv 0$ (see Appendix). 
A simple non-trivial choice is
a second-order tensor taking the form of Eq. \eqref{eq:encoder_model}, so that $\mb B_i(\mb r)$ satisfies
$\mb B_i(\mb Q\mb r) = \mb Q \mb B_i(\mb r) \mb Q^T$ and rotational symmetry of $\mb G(\cdot)$ can be
imposed accordingly.

Model \eqref{eq:moment_model} and \eqref{eq:encoder_model} aims at extracting a set of configuration ``features'', 
represented by the micro-scale encoder $\left\{\mb f_i\right\}_{i=1}^n$ and 
the macro-scale conformation tensor $\left\{\mb c_i\right\}_{i=1}^n$, such
that the polymer stress $-\langle \mb r\nabla V_{\rm b}(r)^{T}\rangle$ can be well approximated by $\mb G(\cdot)$ and the evolution of
 $\{\mb c_i \}_{i=1}^{n}$ can be modeled by $\{\mb H_i(\cdot)\}_{i=1}^n$ self-consistently. 
As a special case, if $n = 1$ and $\mb f_1(\mb r) = \mb r$, $\mb c_1$ recovers the end-end orientation tensor  and
the stress model recovers the aforementioned rheological models under special choices of $\mb G(\cdot)$. 
In practice, to accurately capture the nonlinear effects in $V_{\rm b}$, multiple nonlinear conformation moments are needed. 

To learn $\mb G(\cdot)$ and $\mb H(\cdot)$, one important constraint comes from rotational symmetry.
Let  $\widetilde{\mb r} = \mb Q \mb r$,  where $\mb Q$ 
is unitary.
We must have
\begin{subequations}
\begin{align}
&\mb f_i(\widetilde{\mb r}) =  \mb Q \mb f_i(\mb r) \label{eq:rotation_B}\\
&\mb G(\widetilde{\mb c}_1, \cdots, \widetilde{\mb c}_n) = \mb Q \mb G(\mb c_1, \cdots, \mb c_n) \mb Q^T \\
&\mb H_i(\widetilde{\mb c}_1, \cdots, \widetilde{\mb c}_n) = \mb Q \mb H_i(\mb c_1, \cdots, \mb c_n) \mb Q^T, 
\label{eq:rotation_G_H}
\end{align}
\label{eq:rotation_B_G_H}
\end{subequations}
where $\widetilde{\mb c}_i = \mb Q \mb c_i \mb Q^T$. 
For the tensor derivative $\mathcal{D}\mb c_i/\mathcal{D} t$,  we should have
\begin{equation}
\widetilde{\frac{\mathcal{D} \mb c_i}{\mathcal{D}t}} = \mb Q \frac{\mathcal{D} \mb c_i}{\mathcal{D}t}\mb Q^T 
\quad i = 1, 2, \cdots, n.
\label{eq:rotation_c_evolution}
\end{equation}
This constraint is satisfied by the various objective tensor derivatives in most existing rheological models, such as the
upper-convected \cite{Oldroyd_Wilson_PRSLA_1950}, the covariant \cite{Oldroyd_Wilson_PRSLA_1950} and 
the Zaremba-Jaumann \cite{Zaremba_1903} derivatives, but these forms are not
suitable for us since they lack the desired accuracy. 
Fortunately these constraints are satisfied automatically if we formulate our macro-scale model
based on the underlying micro-scale model. 

We start from the Fokker-Planck equation \cite{Bird_Curtiss_book_vol_2}, 
\begin{equation}
\frac{\partial \rho(\mb r, t)}{\partial t} = -\nabla\cdot
\left[(\bm\kappa\cdot\mb r)\rho - \frac{2k_BT}{\gamma}\nabla\rho 
- \frac{2}{\gamma}\nabla V_{\rm b}(r) \rho\right],
\label{eq:FK_dumbbell} 
\end{equation}
where $k_BT$ is the thermal energy, $\gamma$ is the solvent friction coefficient and $\bm\kappa := \nabla \mb u^T$ is 
the strain of the fluid.  
Instead of solving \eqref{eq:FK_dumbbell}, we consider the evolution of $\mb c_i$, 
\begin{equation}
\frac{\rm d}{\rm dt} \mb c_i - \bm\kappa:\left\langle
\mb r\nabla_{\mb r}\otimes \mb B_i(\mb r) \right\rangle  = 
\frac{2k_BT}{\gamma}\left\langle \nabla^2_{\mb r} \mb B_i(\mb r) \right\rangle
+ 
\frac{2}{\gamma} \left\langle \nabla V_{\rm b}(r) \cdot \nabla_{\mb r} \mb B_i(\mb r) \right\rangle,
\label{eq:FK_B_evoluation}
\end{equation}
where $:$ is the double-dot product. We can prove that Eqs. (\ref{eq:FK_dumbbell})
and (\ref{eq:FK_B_evoluation}) are rotationally invariant. In particular, the combined
left-hand-side terms of \eqref{eq:FK_B_evoluation} satisfy the symmetry condition 
in \eqref{eq:rotation_c_evolution} (see proof in Appendix). 
Therefore, the combined terms establish a generalized objective
tensor derivative. It takes a different form from the ones \cite{Oldroyd_Wilson_PRSLA_1950, Zaremba_1903} in existing models
and rigorously preserves the rotational symmetry condition (\ref{eq:rotation_c_evolution}). Accordingly, 
the hydrodynamic model (\ref{eq:moment_model}) take the following ansatz
\begin{subequations}
\begin{align}
\frac{\mathcal{D} \mb c_i}{\mathcal{D}t} &= \frac{\rm d}{\rm dt} \mb c_i - \bm\kappa:\mathcal{E}_i 
\label{eq:c_tensor_def} \\
\mb H_i &= \frac{2k_BT}{\gamma} \mb H_{1,i} + \frac{2}{\gamma} \mb H_{2,i}.
\end{align}
\label{eq:c_evolution_ansatz}
\end{subequations}
Each term of \eqref{eq:c_evolution_ansatz} has a micro-scale correspondence, i.e., 
\begin{equation}
\begin{split}
&\mathcal{E}_i(\mb c_1, \cdots, \mb c_n) =  \left\langle \mb r\nabla_{\mb r}\otimes \mb B_i(\mb r) \right\rangle  \\
&\mb H_{1,i} (\mb c_1, \cdots, \mb c_n) = \left\langle \nabla^2_{\mb r} \mb B_i(\mb r) \right\rangle  \\
&\mb H_{2,i} (\mb c_1, \cdots, \mb c_n) = \left\langle \nabla V_{\rm b}(r) \cdot \nabla_{\mb r} \mb B_i(\mb r) \right\rangle,
\end{split}
\label{eq:E_H_DNN}
\end{equation}
where $\mathcal{E}_i$ is a $4\mhyphen$th order tensor function and $\mb H_{1,i}$, $\mb H_{2,i}$ are $2$nd order 
tensor functions. They will be approximately represented by DNNs. To collect the training data, we use microscale simulations to evaluate
these terms; no time-series samples are needed.

Note that  $\mathcal{D}\mb c_i/\mathcal{D}t$ depends on  $\mathcal{E}_i$, which 
encodes some  micro-scale information from $\mb B_i(\mb r)$. 
Different from the common choices of the objective tensor derivatives in existing models, it takes a more
general formulation and has a \emph{clear} micro-scale interpretation without the conventional 
ambiguities. It recovers the standard upper-convected derivative under
special case. To the best of our knowledge, this is the first study that
establishes such a direct micro-scale linkage for the objective tensor derivative in the non-Newtonian 
fluid modeling. The present form is different from the common choices
of the objective tensor derivatives in existing models. As shown later, such a formulation 
that faithfully accounts for the micro-scale polymer configuration 
is crucial for the accuracy of the constitutive model for $\mb c_i$. 

%


\subsection{Symmetry-preserving DNN representations}
Special rotation-symmetry-preserving DNNs are needed for 
the encoder functions $\left\{\mb f_i\right\}_{i=1}^{n}$, the $2$nd order 
tensors $\mb G$ and $\left\{\mb H_{1,i}, \mb H_{2,i}\right\}_{i=1}^n$ and the $4$th order tensors $\left\{\mathcal{E}_i\right\}_{i=1}^{n}$ such that
the symmetry conditions \eqref{eq:rotation_B_G_H} and \eqref{eq:rotation_c_evolution} are satisfied. 
For Eq. \eqref{eq:rotation_B} to hold, one can show that $\mb f_i(\mb r)$ has to take the form 
\begin{equation}
\mb f_i(\mb r) = g_i(r)\mb r,
\label{eq:encoder_DNN}  
\end{equation}
where $g_i(r)$ is a scalar encoder function (see Appendix).
We always set $g_1(r) \equiv 1$, yielding $\mb G \propto \mb H_{2,1}$.

To construct the DNNs for $\mb G$ and $\left\{\mb H_{1,i}, \mb H_{2,i}\right\}_{i=1}^n$ that satisfy
Eq. \eqref{eq:rotation_G_H}, we can transform $\left\{\mb c_i\right\}_{i=1}^n$ 
into a fixed frame for the DNN input. 
One natural choice is the eigen-space of the conformation tensor $\mb c_1= \langle \mb r\mb r^T\rangle$.
Let $\mb V$ be the matrix composed of the eigenvectors of $\mb c_1$.
Define
\begin{equation}
\begin{split}
&\mb H_{j,i}(\mb c_1, \cdots \mb c_n) = \mb V \widehat{\mb H}_{j,i}(\widehat{\mb c}_1, \cdots \widehat{\mb c}_n) \mb V^T\\
&\widehat{\mb c}_i = \mb V^T \mb c_i \mb V \quad j = 1,2~~i = 1, \cdots, n, 
\end{split}
\label{eq:G_H_ansatz}
\end{equation}
$\widehat{\mb c}_1$ is a diagonal matrix composed of the eigenvalues of $\mb c_1$.
The DNNs will be constructed to learn $\widehat{\mb H}_{j,i}$.

The learning of the $4$th order tensors $\left\{\mathcal{E}_i\right\}_{i=1}^{n}$ is based on the following decomposition:
\begin{equation}
\begin{split}
&\mathcal{E}_i(\mb c_1,\cdots, \mb c_n)
= \left\langle g_i(r)^2 \mb r\nabla_{\color{black}{\mb r}} 
\otimes
\color{black}{\mb r} \color{black}{\mb r^T} \right\rangle \\
&+ \sum_{k=1}^6 \mb E_{1,i}^{(k)} (\mb c_1, \cdots, \mb c_n) \otimes \mb E_{2,i}^{(k)} (\mb c_1, \cdots, \mb c_n).
\end{split}
\label{eq:E_tensor_ansatz}
\end{equation}
$\mb E_{1,i}, \mb E_{2,i} \in \mathbb{R}^{3\times3}$ are second-order tensors satisfying the rotational symmetry condition similar to 
Eq. \eqref{eq:rotation_G_H}, i.e., 
\begin{equation}
\begin{split}
&\mb E_{1,i}(\tilde{\mb c}_1, \cdots, \tilde{\mb c}_n) = \mb Q \mb E_{1,i} (\mb c_1, \cdots, \mb c_n) \mb Q^T \\
&\mb E_{2,i}(\tilde{\mb c}_1, \cdots, \tilde{\mb c}_n) = \mb Q \mb E_{2,i} (\mb c_1, \cdots, \mb c_n) \mb Q^T.
\end{split}
\label{eq:E_F_symmetry}
\end{equation}
It is shown in Appendix that this decomposition satisfies Eq. \eqref{eq:rotation_c_evolution}.
Accordingly, $\mathcal{E}_i$ can be constructed by a set of second order tensors $\mb E_{1,i}$ and $\mb E_{2,i}$, which can be
constructed similar to Eq. \eqref{eq:G_H_ansatz}.
Note that with this form, the first term in the RHS of Eq. \eqref{eq:E_tensor_ansatz} becomes
 $\bm\kappa \mb c_i + \mb c_i \bm\kappa ^T$ similar to the upper-convected derivative.   
In summary, the DNNs are designed to parametrize
$\{g_i(r)\}_{i=2}^n, \{\widehat{\mb H}_{1,i}, \widehat{\mb H}_{2,i}, \widehat{\mathcal{E}}_i \}_{i=1}^n$. 

Finally, the DNNs are trained by minimizing the loss
\begin{equation}
L = \lambda_{H_1} L_{H_1} + \lambda_{H_2} L_{H_2} + \lambda_{\mathcal{E}} L_{\mathcal{E}},
\end{equation}
where $L_{H_1}$, $L_{H_2}$ and $L_{\mathcal{E}}$ are the empirical risk associated with
$\left\{\mb H_{1,i}\right\}_{i=1}^n$, $\left\{\mb H_{2,i}\right\}_{i=1}^n$ and $\left\{\mathcal{E}_i\right\}_{i=1}^n$ 
respectively.    
$\lambda_{H_1}$, $\lambda_{H_2}$ and $\lambda_{\mathcal{E}}$ are hyper-parameters (see Appendix). Note that the 
encoders $\left\{g_i(r)\right\}_{i=2}^n$ do not explicitly appear in 
$L$; they are trained through the learning of $\widehat{\mb H}$ and $\mathcal{E}$. 

\subsection{DeePN$^2$}
The DeePN$^2$ model is made up of 
Eqs. \eqref{eq:momentum_transport_close}, \eqref{eq:moment_model} and \eqref{eq:c_evolution_ansatz}.
Note that the model takes the form of classical empirical models. The only differences are that some new conformation tensors
and a new form of objective tensor derivative are introduced, and some of the equation terms are represented as 
function subroutines in the form of NN models.  The latter is no different from the situation commonly found in
gas dynamics \cite{Molecular_gas_bird_1994}, where the equations of state 
are given as tables or function subroutines.
Also, we note that such conformation tensors are learned from the micro-scale simulations for the best approximation
of the polymer stress and constitutive dynamics. This allows us to bypass the evaluation of the polymer configuration
distribution by directly solving the high-dimensional FK equation, or coupling the micro-scale simulations. Meanwhile, the 
micro-scale viscoelastic effects can be faithfully captured beyond the empirical closures based on the linear/mean-field approximations. 

\section{Numerical results}
To demonstrate the model accuracy, we consider a polymer solution with polymer number density $n_p = 0.5$. The bond 
potential $V_{\rm b}(r)$ is chosen to be FENE, i.e., 
$V_{\rm b}(r)
=-\frac{k_s}{2}r^2_{0}
\log\left[ 1-\frac{r^2}{r^2_{0}}\right]$,
where $k_s$ is the spring constant.  $r_0$ is the maximum bond extension.
The continuum model is constructed using $n=3$ encoder functions.
We also experimented with larger values of $n$ but did not
see appreciable improvement. 
That been said, the choice of $n$ 
needs to be looked into more carefully in the future.

\begin{figure*}
\centering
\includegraphics[trim=60 20 100 70,clip,scale=0.21]{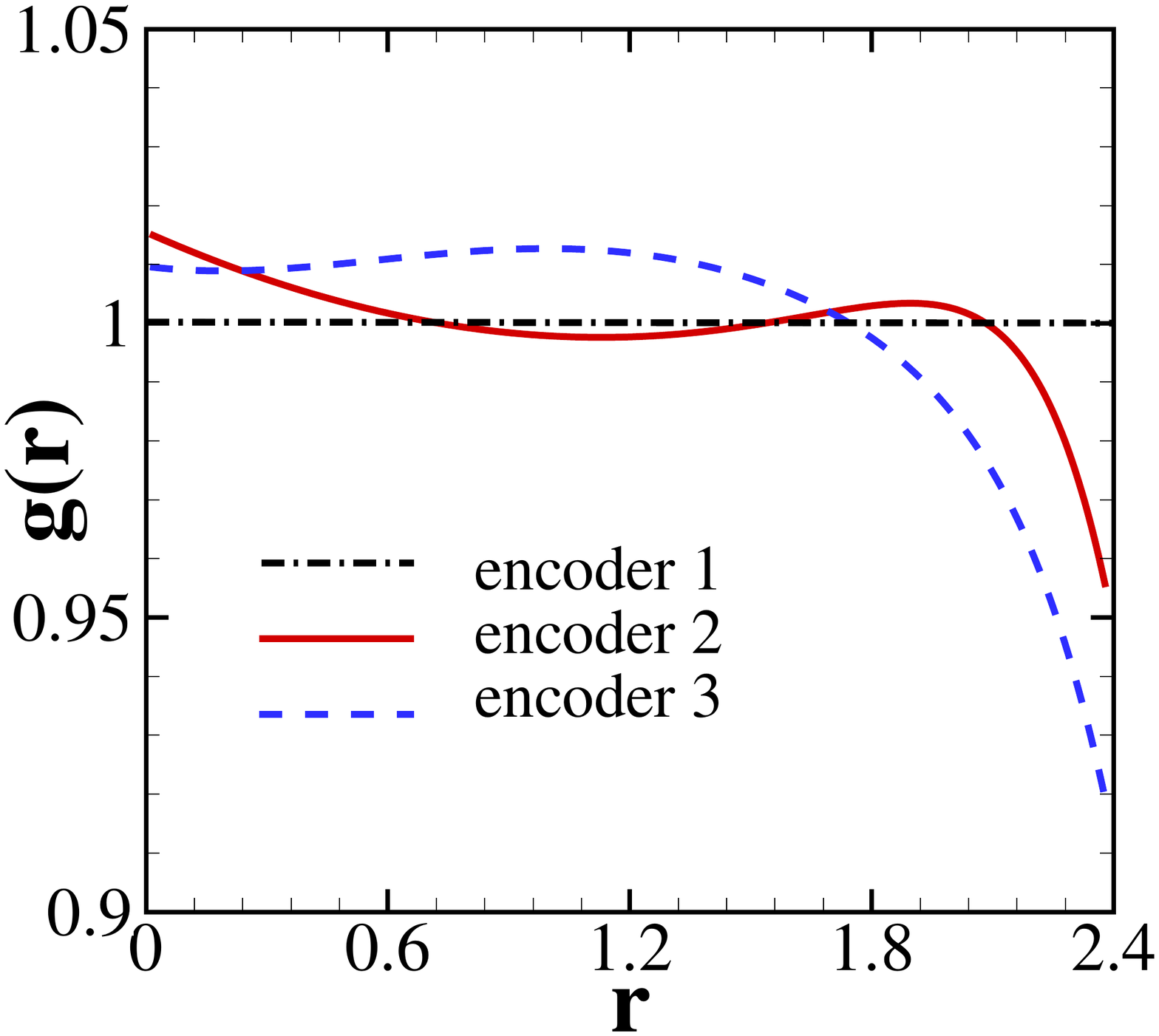}
\includegraphics[trim=60 20 100 70,clip,scale=0.21]{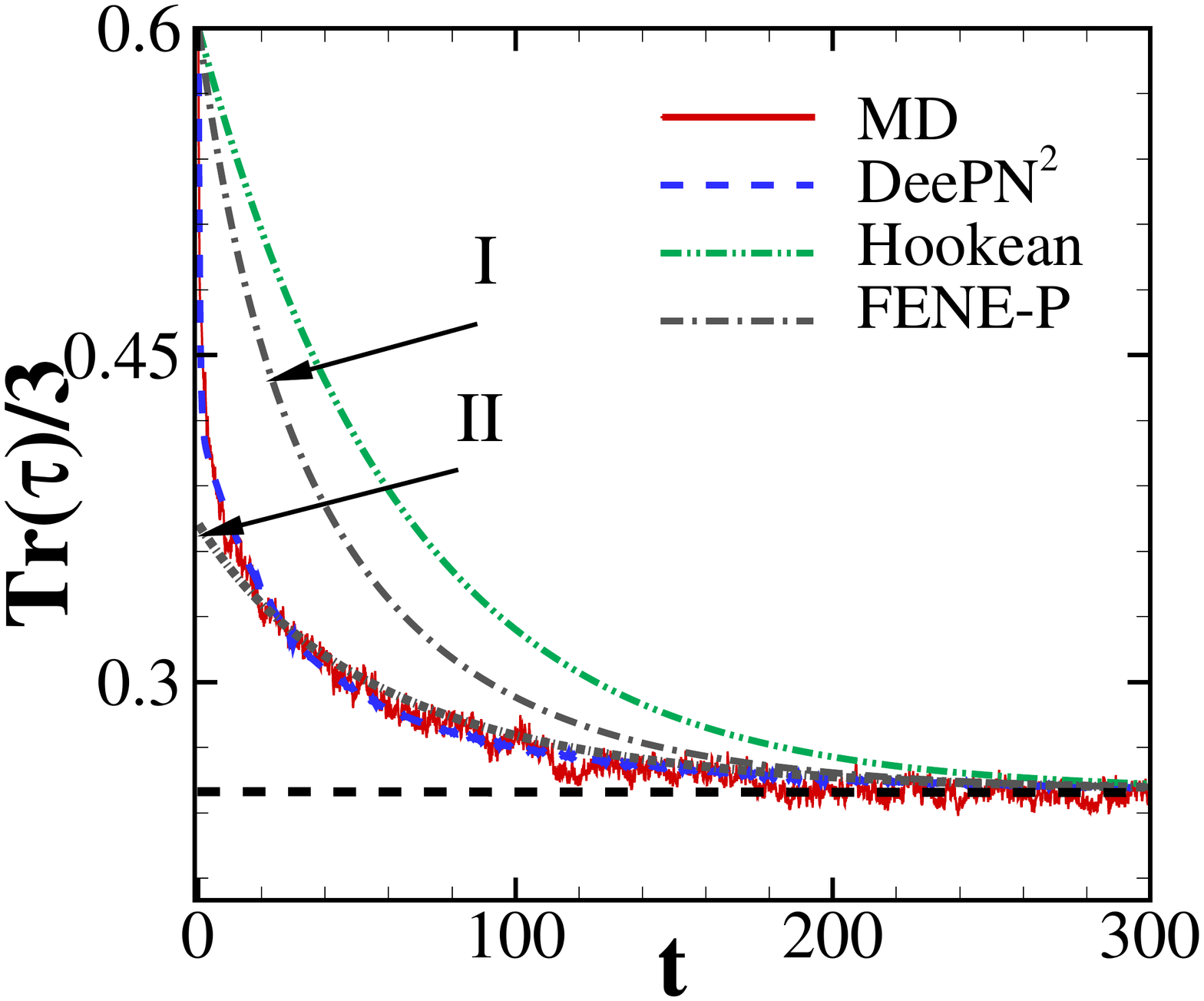}
\includegraphics[trim=60 20 100 70,clip,scale=0.21]{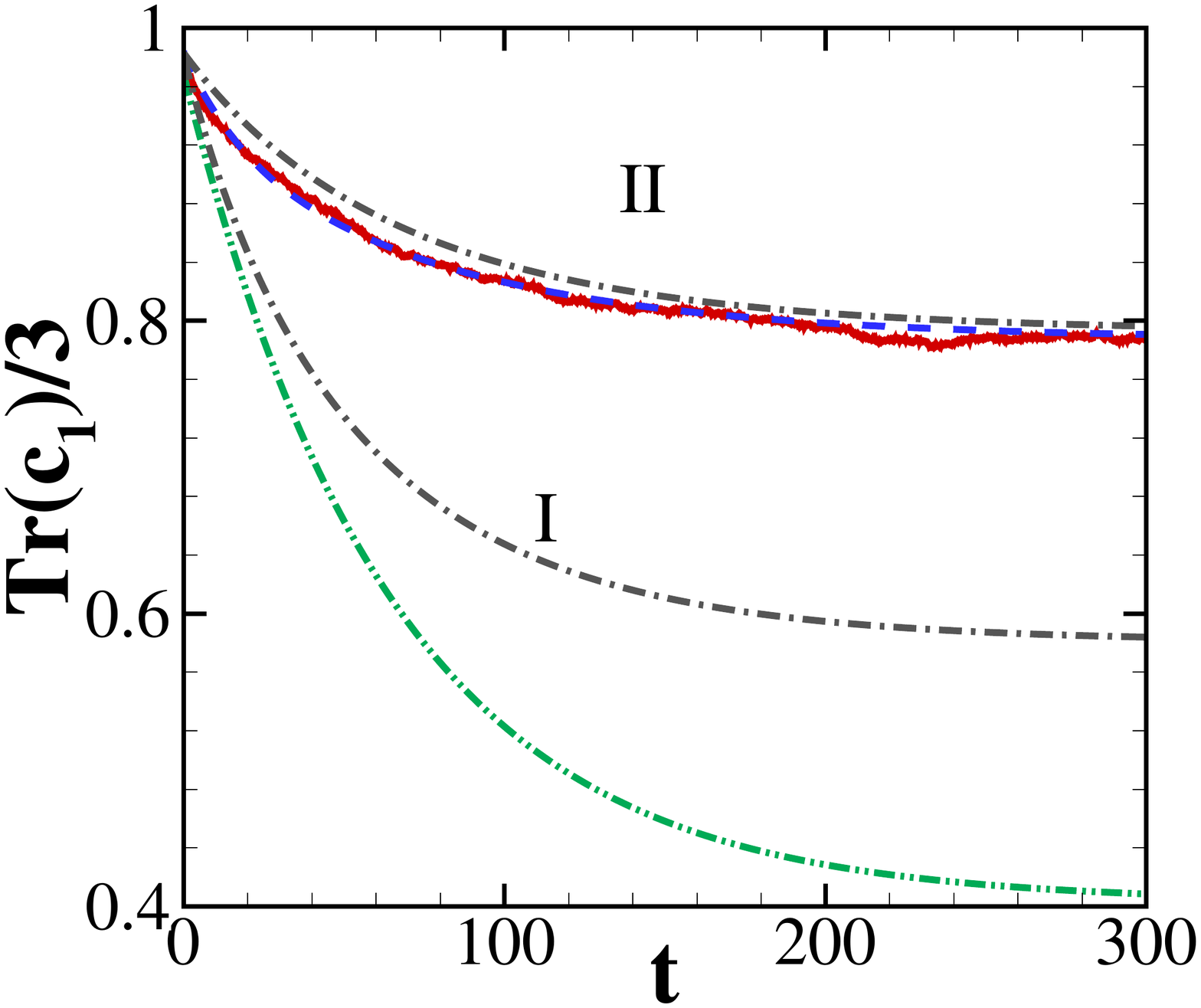}
\caption{Quasi-equilibrium relaxation process of a dumbbell suspension obtained from direct MD simulation, the present 
DeePN$^2$, canonical 
Hookean and FENE-P model. \textbf{Left}: Encoder function $g(r)$. 
\textbf{Middle}: Evolution of $\bm \tau_{\rm p}$. \textbf{Right}: Evolution of $\mb c_1 = \left\langle \mb r\mb r^T\right\rangle$.
Two parameter sets of 
the FENE-P model are examined: (\Rmnum{1}) the initial conditions of both $\mb c_1$ 
and $\bm\tau_{\rm p}$  are consistent with MD. (\Rmnum{2}) the initial and final conditions
of $\mb c_1$ are consistent with MD. The Hookean model parameters are 
set following (\Rmnum{1}).
}
\label{fig:quasi_equilibrium_MD_ML}
\end{figure*}

Fig. \ref{fig:quasi_equilibrium_MD_ML} shows the encoder functions  $g(r)$  with $r_0 = 2.4$, $k_s = 0.1$. 
To validate the DeePN$^2$ model, we consider a quasi-equilibrium dynamics of the polymer solution with
$k_B T = 0.25$, while the initial polymer configuration is taken from the equilibrium state with $k_B T = 0.6$. The
relaxation process is simulated using both  MD and DeepN$^2$. Fig. \ref{fig:quasi_equilibrium_MD_ML} shows the evolution of 
the trace  of $\mb c_1$ and $\bm\tau_{\rm p}$. The predictions from DeePN$^2$ agree well with the MD results. 
In contrast, the predictions from Hookean and FENE-P model show apparent deviations.

Next, we consider the non-equilibrium process of a reverse Poiseuille flow (RPF) in a domain 
$[0, 40]\times[0, 80]\times [0, 40]$ (reduced unit), with periodic boundary condition imposed in each direction. 
Starting from $t = 0$, an external field $\mb f_{\rm ext} = (f_b, 0, 0)$ 
is applied in 
the region $y < 40$ and an opposite field $\mb f_{\rm ext} = (-f_b, 0, 0)$ is applied 
in the region $y > 40$. 
Fig. \ref{fig:velocity_profile_conformation} shows the instantaneous velocity profiles with $r_0 = 3.8$ and $f_b = 0.02$. The predictions from
DeePN$^2$ agree well with MD while FENE-P yields apparent deviations.
For the velocity evolution at 
$y = 6$ and $y = 14$, the predictions from the Hookean and FENE-P models show 
pronounced overestimations on both the magnitude and duration of the oscillation behavior. Such limitations of the FENE-P
model have already been noted in Ref. \cite{Laso_Ottinger_JNNFM_1993}. 
From the microscopic perspective, the discrepancy arises from the 
mean field approximation, $\bm\tau_{\rm p} \approx \mb c/(1 - \Tr(\mb c)/r_0^2)$. 
Such an approximation 
cannot capture the nonlinear response when individual polymer bond length approaches $r_0$.
In contrast, DeePN$^2$ 
can capture such micro-scale ``bond length dispersion'' via the additional macro-scale nonlinear conformation 
tensors $\mb c_2, \cdots, \mb c_n$.

\begin{figure}[htbp]
\includegraphics[trim=60 20 100 70,clip,scale=0.25]{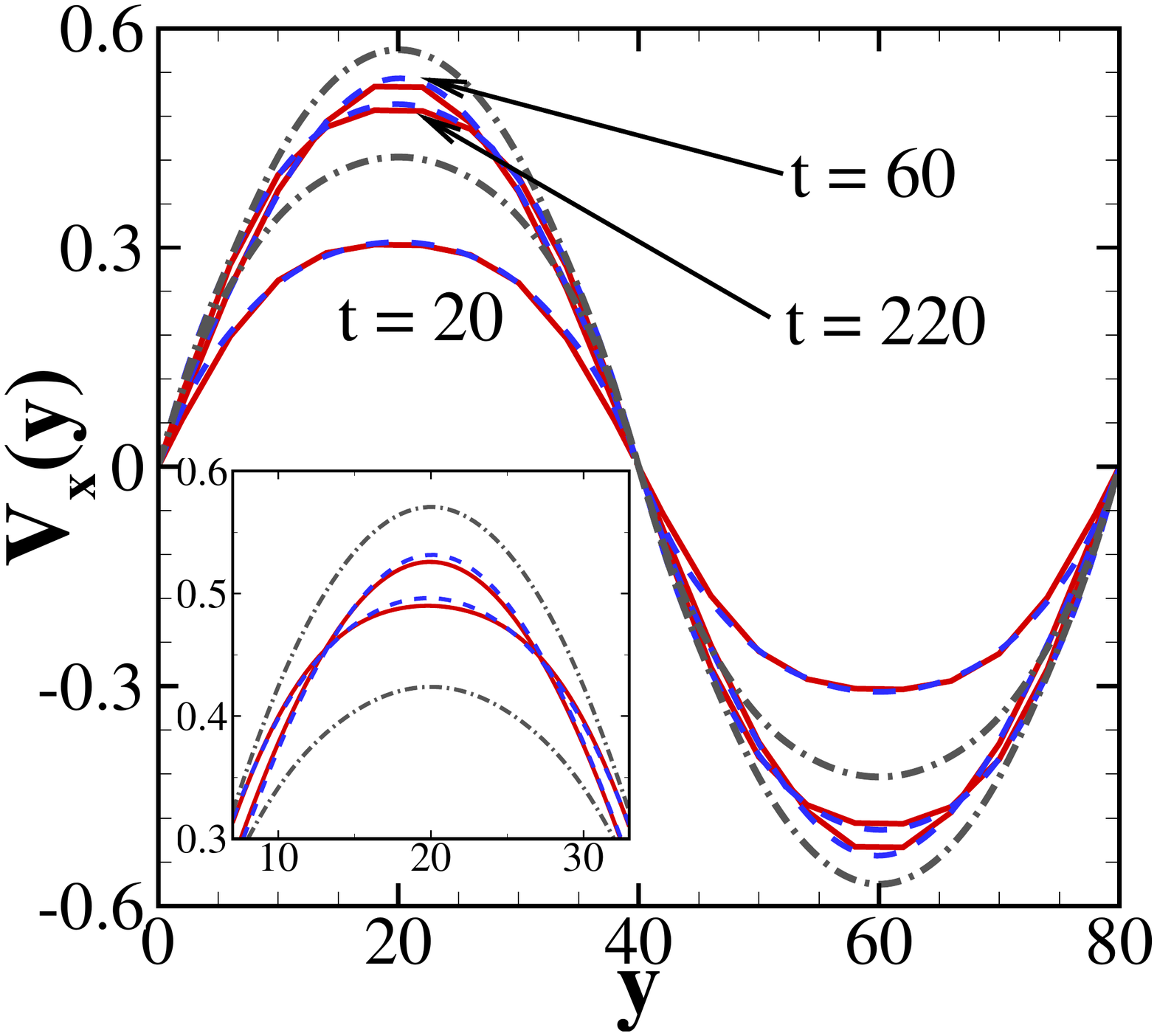}
\includegraphics[trim=60 20 100 70,clip,scale=0.25]{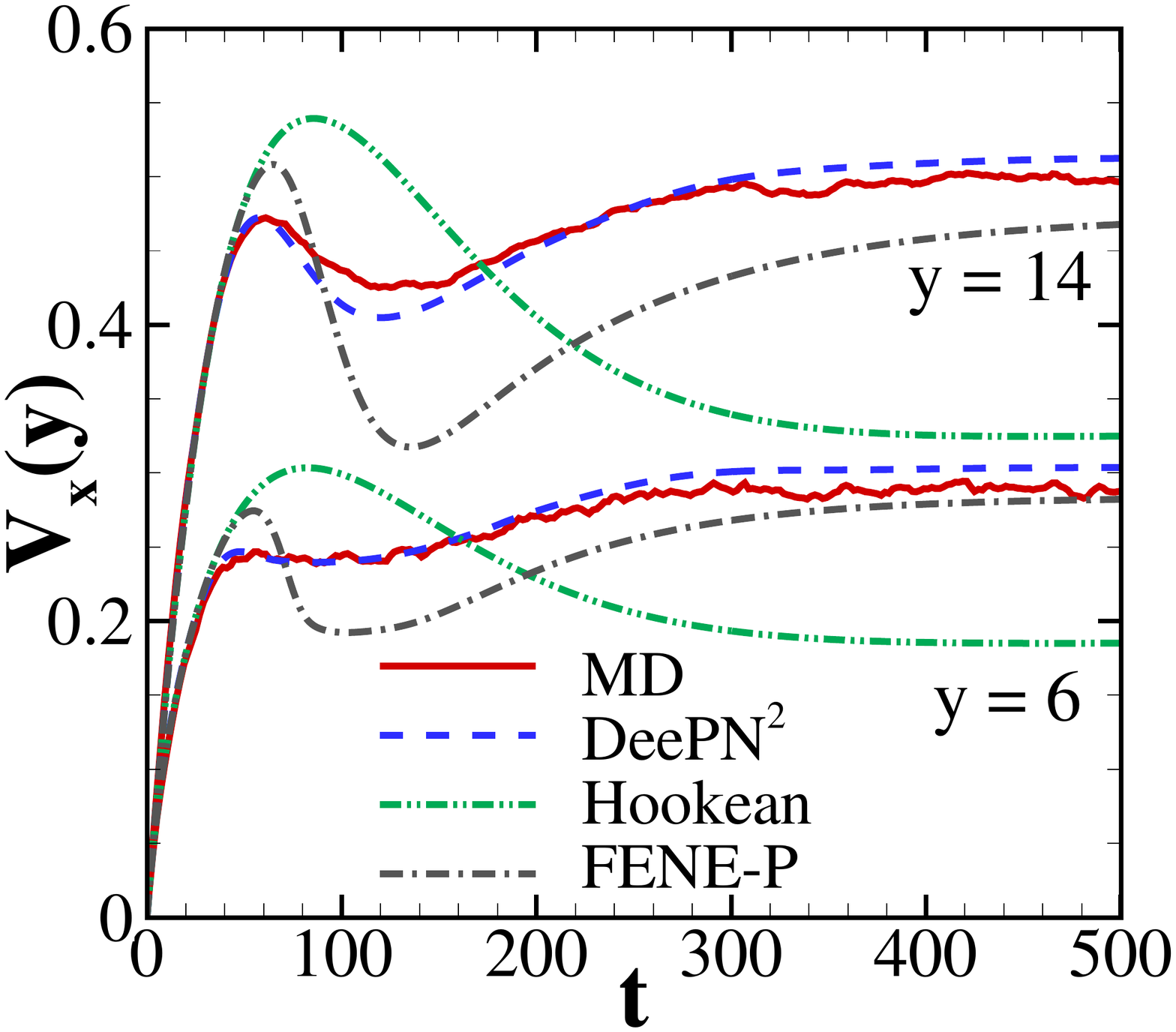}
\caption{Evolution of the reverse Poiseuille flow of a dumbbell suspension obtained from MD and various models. \textbf{Left}:
velocity profiles at $t = 20, 60, 220$. \textbf{Right}: velocity evolution at $y = 6$ and $y = 14$. 
The parameters of the Hookean and FENE-P model are chosen 
such that the equilibrium bond length matches the MD results. 
}
\label{fig:velocity_profile_conformation}
\end{figure}

Shown in  Fig. \ref{fig:micro_macro_link}(a) is the evolution of $\mb c_1$ at $y = 6$. 
The DeePN$^2$ faithfully  predicts the responses of the polymer configurations 
under the external flow field. The instantaneous $\bm \tau_{\rm p}$ is also accurately
predicted by the conformation tensors, 
as shown in Fig. \ref{fig:micro_macro_link}(b-c). 
The responses 
can also be examined by 
the shear-rate-dependent viscosity. As shown in Fig. \ref{fig:micro_macro_link}(d), predictions 
by DeePN$^2$ agree well with the MD results.  
In contrast, the 
FENE-P model 
yields apparent deviations. 

\begin{figure}[htbp]
\includegraphics[trim=60 20 100 20,clip,scale=0.25]{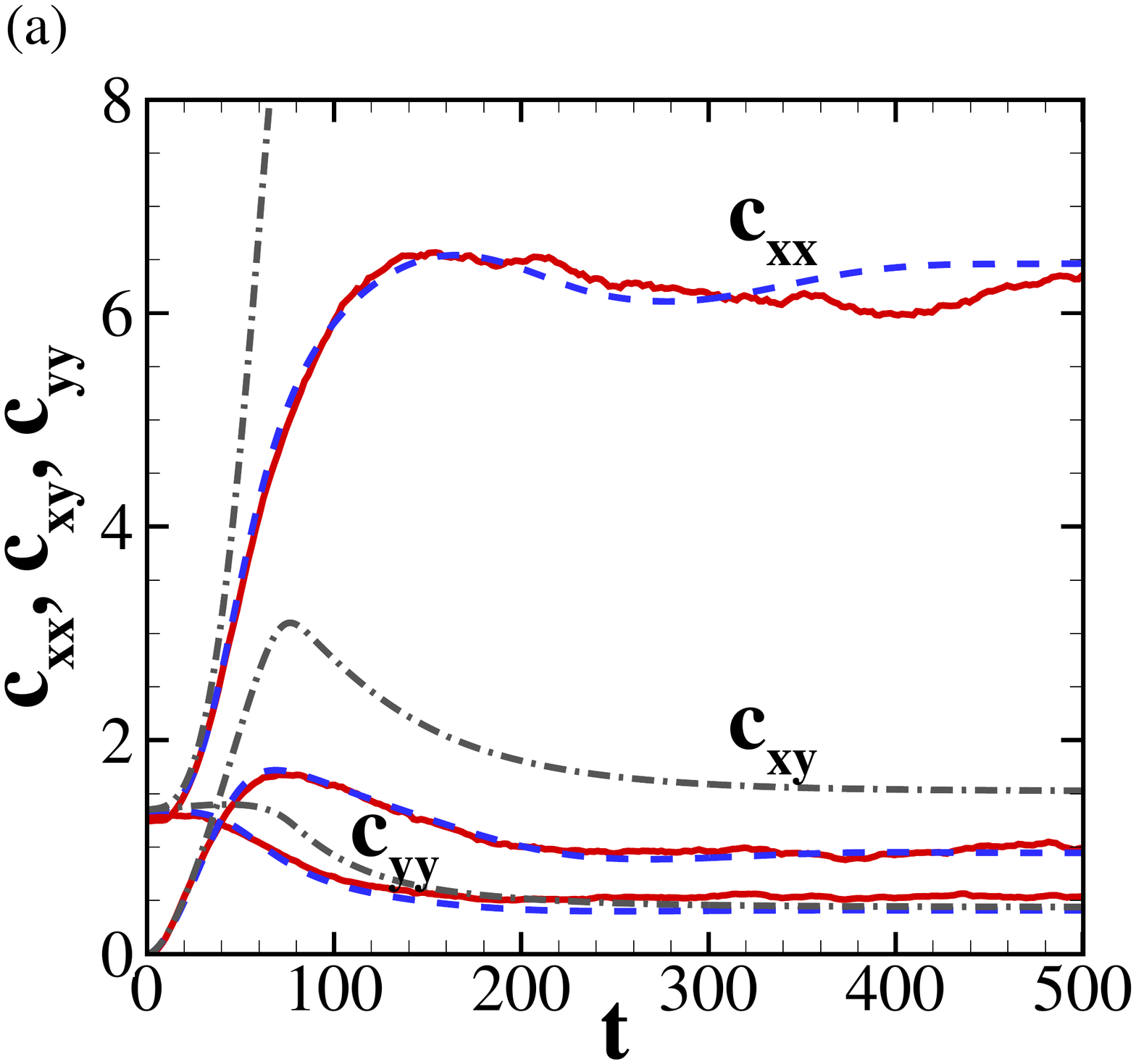}
\includegraphics[trim=60 20 100 20,clip,scale=0.25]{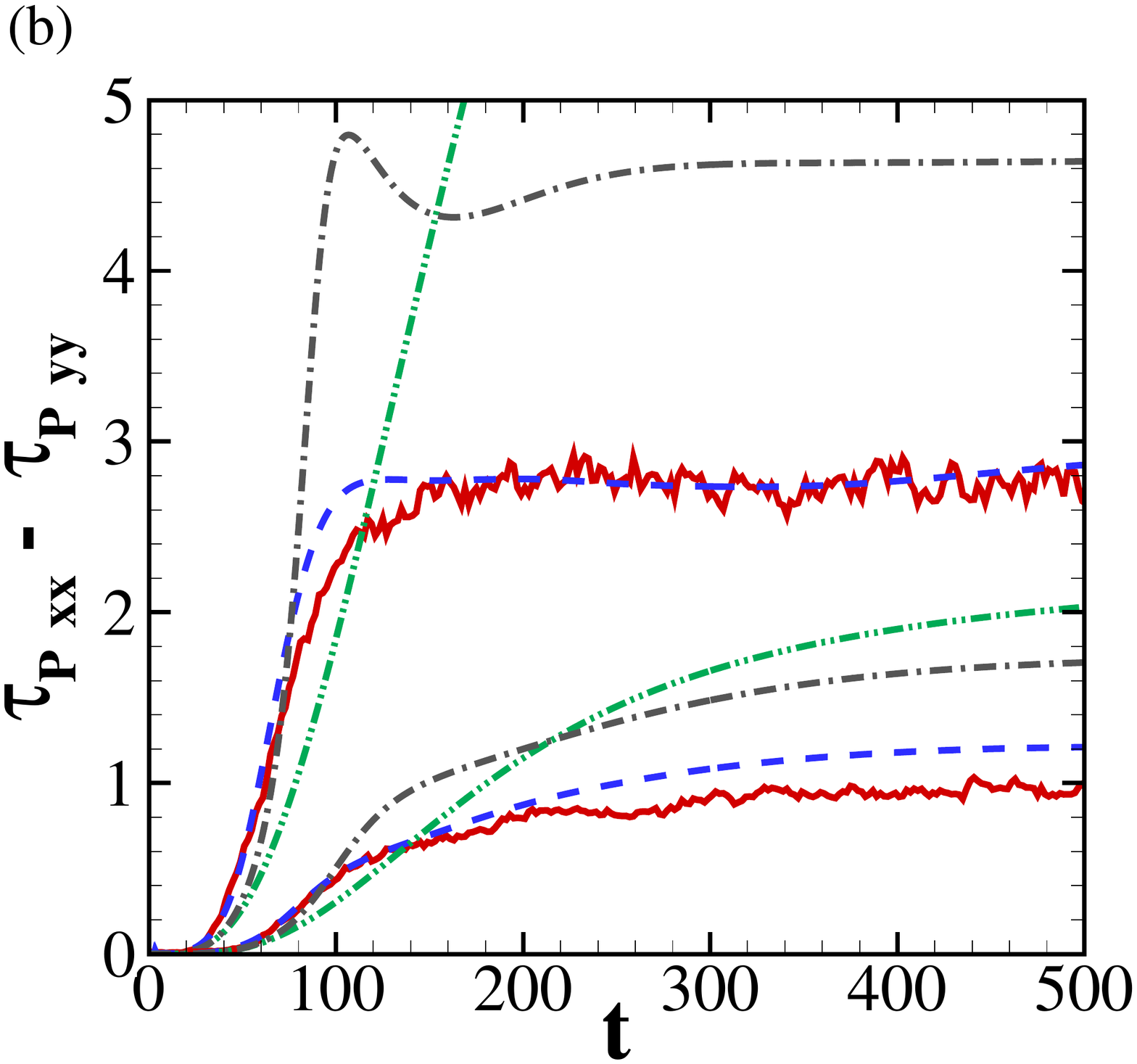}\\
\includegraphics[trim=60 20 100 20,clip,scale=0.25]{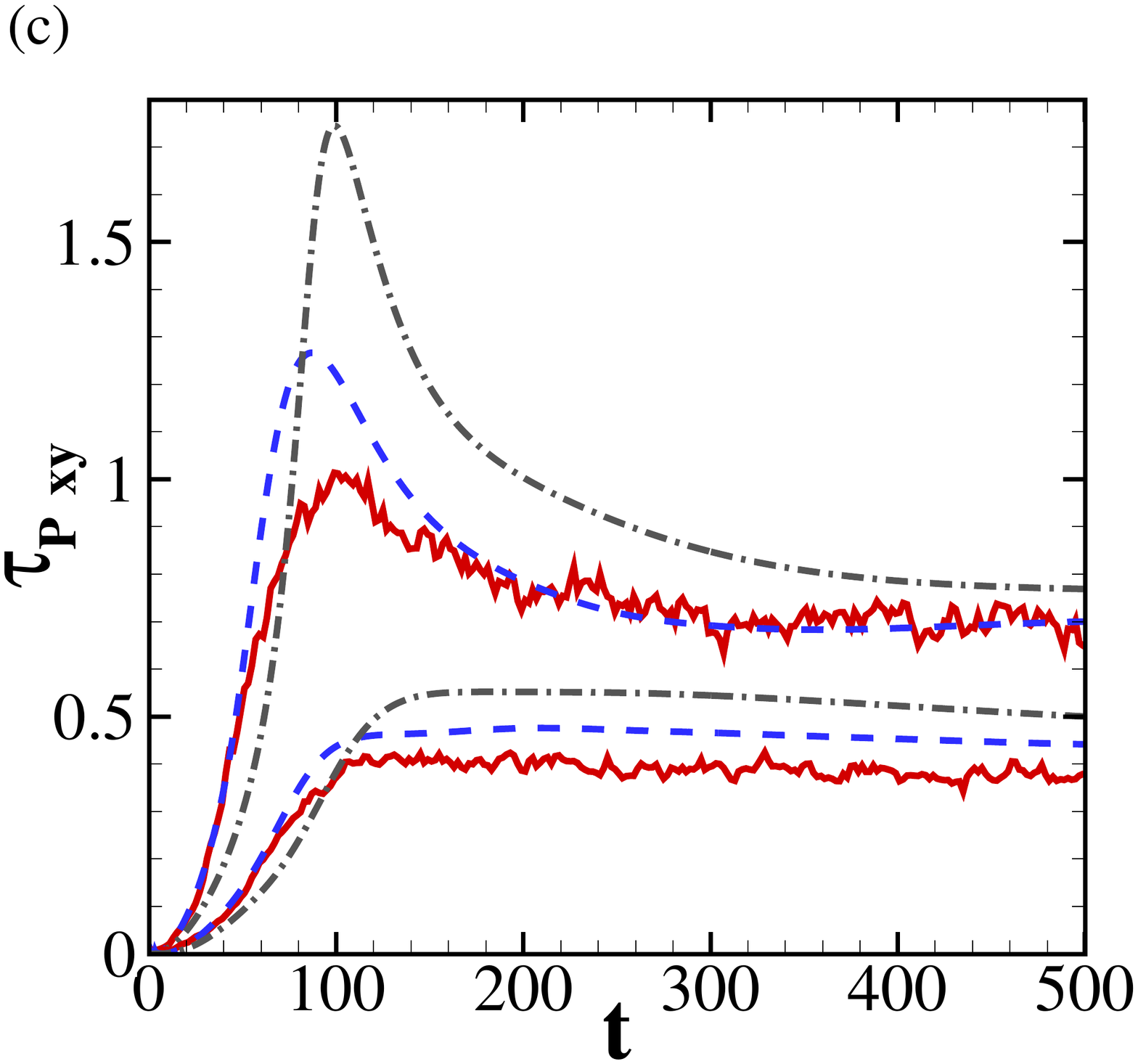} 
\includegraphics[trim=60 20 100 20,clip,scale=0.25]{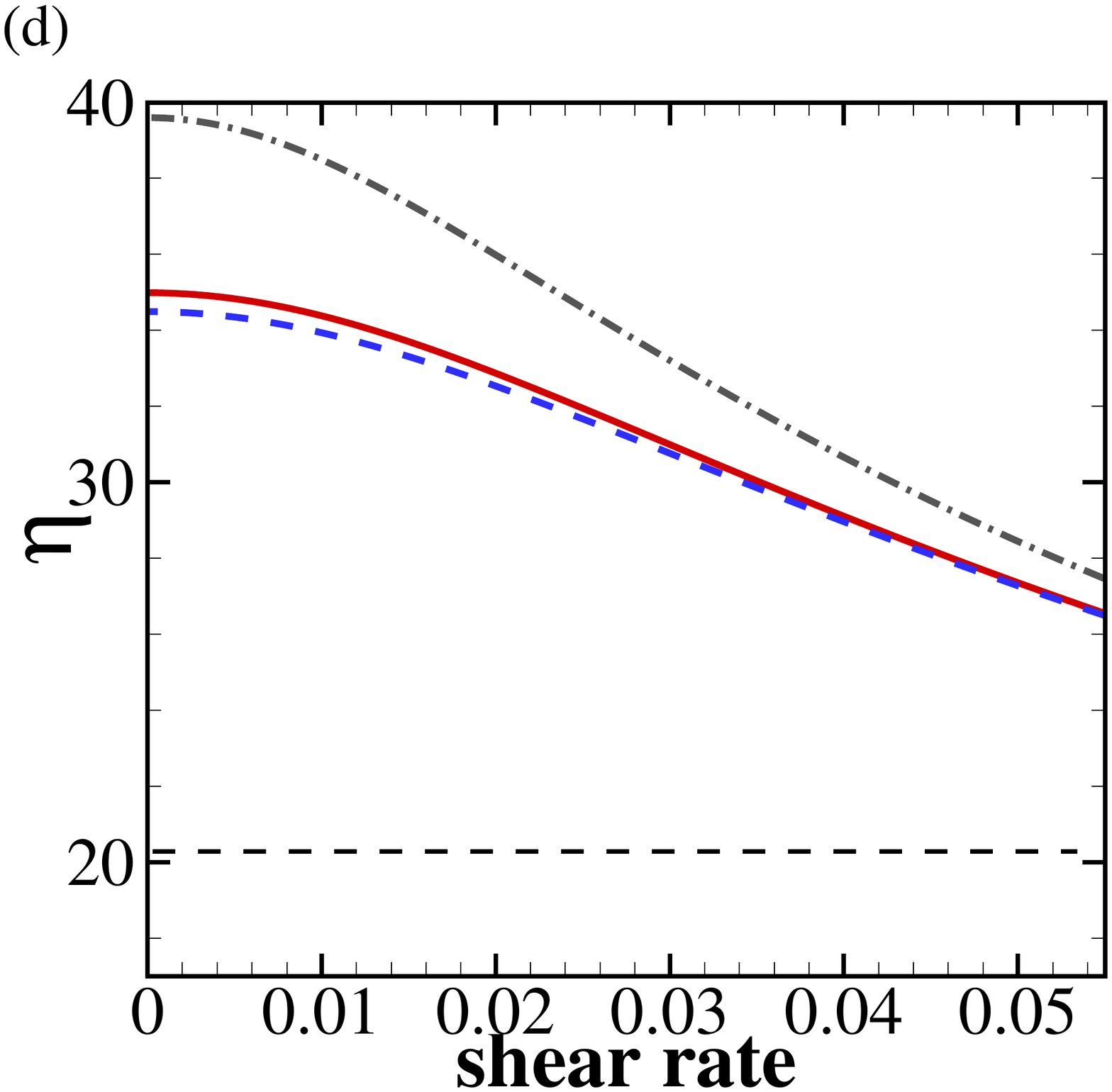}
\caption{The micro-macro correspondence during the evolution of the reverse Poiseuille flow 
of the dumbbell suspension presented in Fig. \ref{fig:velocity_profile_conformation} with the 
same line scheme. 
(a) evolution of $\mb c_1$ at $y = 6$. (b-c) 
Normal stress difference ${\bm \tau_{\rm p}}_{xx}- {\bm\tau_{\rm p}}_{yy}$ and shear stress ${\bm \tau_{\rm p}}_{xy}$ 
at $y = 6$ (upper lines) and $y = 14$ (lower lines). (d) Shear-rate-dependent viscosity. 
The predictions by Hookean model show large deviations from the MD results, and not shown in
(a), (c), (d) for visualization purpose.}
\label{fig:micro_macro_link}
\end{figure}

Besides the first-principle-based stress model and dynamic closure, another distinctive feature of the  
DeePN$^2$ model is the generalized objective tensor derivative $\mathcal{D}\mb c_i/\mathcal{D}t$: 
\begin{equation}
\frac{\mathcal{D}\mb c_i}{\mathcal{D}t} = \overset{\triangledown}{\mb c}_i - 
\bm\kappa:\left[\sum_{k=1}^6 \mb E_{1,i}^{(k)} (\mb c_1,\cdots,\mb c_n) \otimes \mb E_{2,i}^{(k)} (\mb c_1,\cdots,\mb c_n)\right],
\label{eq:tensor_derivative}
\end{equation}
where $\overset{\triangledown}{\mb c}_i$ is the standard upper-convected derivative and 
the second term arises from the source term 
$\left\langle \mb r \nabla_{\mb r} g(r)^2 \otimes \mb r\mb r^T \right\rangle$
in  Eq. \eqref{eq:FK_B_evoluation}.
Therefore, the second term of $\mathcal{D}\mb c_i/\mathcal{D}t$ is embedded with the nonlinear 
response to external field $\bm\kappa$, 
inherited from the encoder $g_i(r)$.  
As a numerical test, we use the present model to simulate the RPF, where  
$\mathcal{D}\mb c_i/\mathcal{D}t$ is chosen to be the upper-convected derivative $\overset{\triangledown}{\mb c}_i$ and 
other modeling terms remain the same. Fig. \ref{fig:corroational_derivative} shows the evolution of the velocities and $\mb c_1$. 
By ignoring the second term in Eq. \eqref{eq:tensor_derivative}, the predictions 
show apparent deviations from the MD results.
This  indicates that the empirical choices of the objective tensor derivative are not accurate.
To achieve the desired accuracy, these derivatives have to 
retain some information from the specific conformation tensor.

\begin{figure}[htbp]
\centering
\includegraphics[trim=60 20 100 20,clip,scale=0.25]{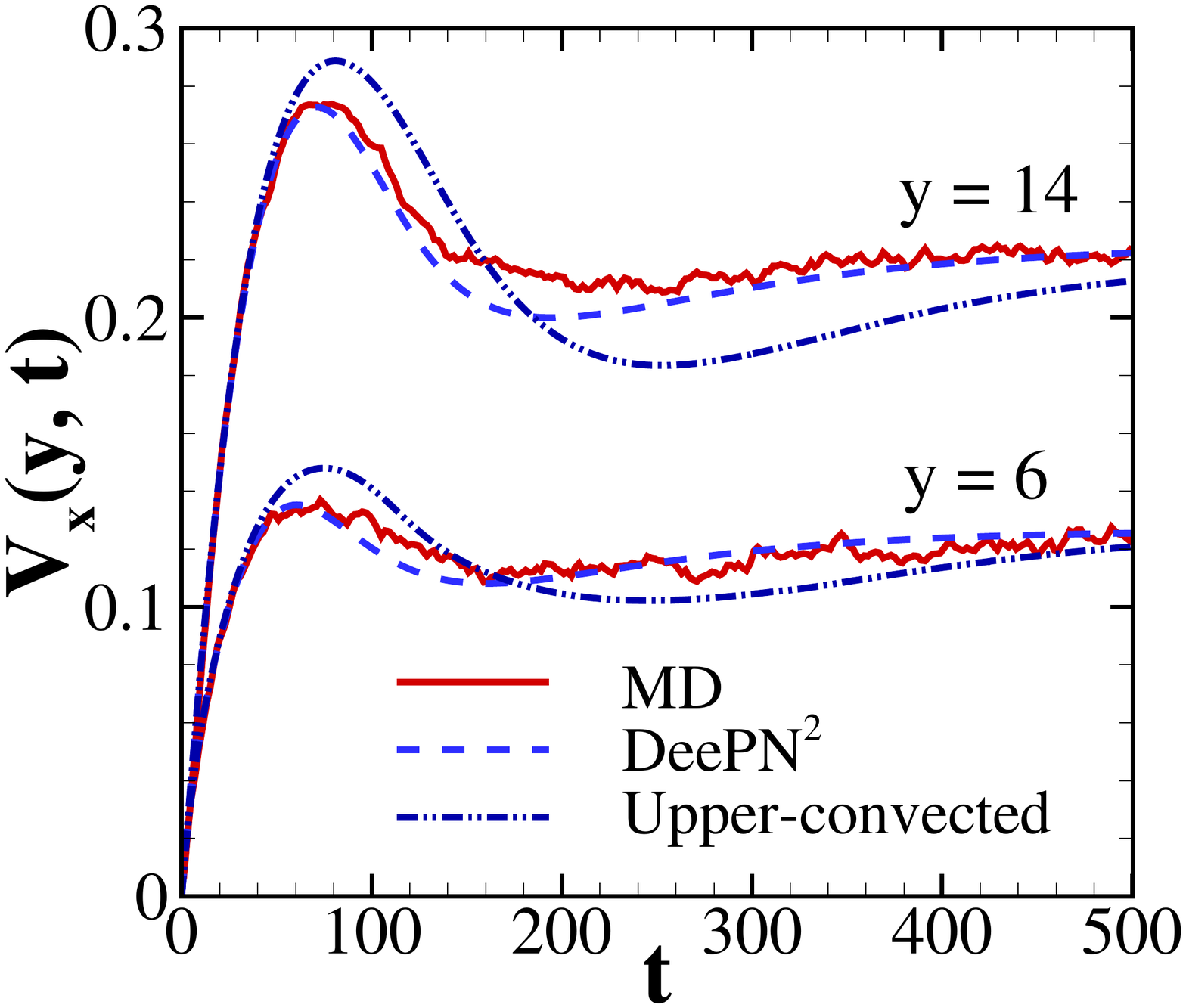}
\includegraphics[trim=60 20 100 20,clip,scale=0.25]{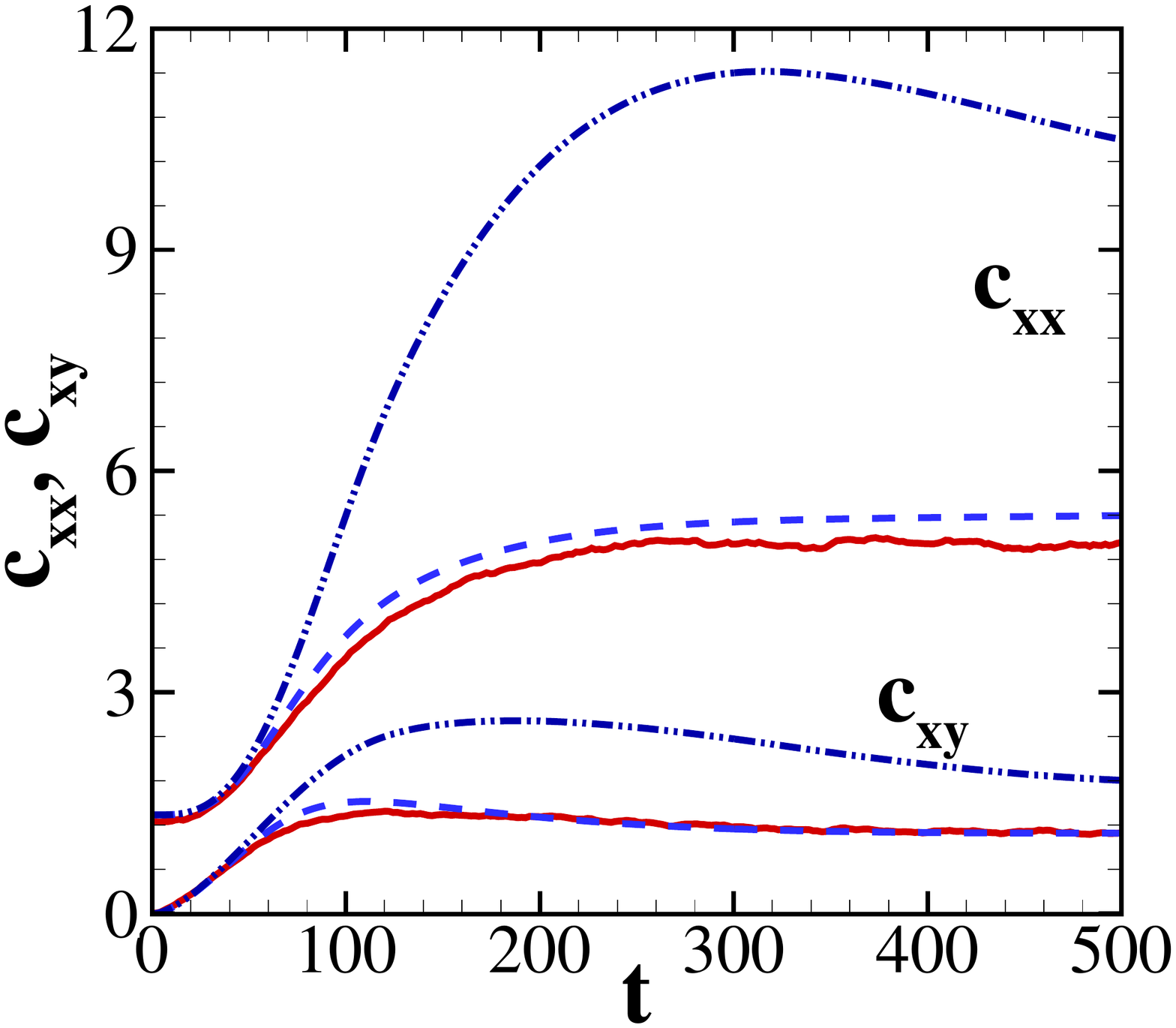}
\caption{The effectiveness of the objective tensor derivative constructed 
by \eqref{eq:tensor_derivative}.  
The additional source term plays a vital role in the accurate modeling of the fluid systems. 
The model that uses the canonical upper-convected
derivative shows apparent deviations from the 
MD results for the evolution of the velocities (\textbf{left}) and 
$\mb c_1$ (\textbf{right}) at $y = 6$.  
}
\label{fig:corroational_derivative}
\end{figure}

\section{Discussion}
The present DeePN$^2$ directly learns the stress model and constitutive dynamics from the microscale simulation data and
avoids dealing with the high-dimensional microscale configuration density function $\rho(\mb r, t)$. 
A main observation is that the explicit knowledge of $\rho(\mb r, t)$ is a sufficient, but not a necessary condition 
for constructing the full constitutive equation. We note that DeePN$^2$ differs from the previous
moment-closure studies \cite{Warner_IECF_1972, Warner_PhD_1971, Armstrong_JCP_1974_1} based on empirical approximations
of $\rho(\mb r, t)$. In these semi-analytical studies \cite{Warner_IECF_1972, Warner_PhD_1971, Armstrong_JCP_1974_1}, the steady-state
FK solution $\rho_{s}(\mb r)$ of a dumbbell is approximated by series expansion, yielding the stress-strain relationship
only for equilibrium \cite{Bird_Curtiss_book_vol_2}. In Ref. \cite{Du_Liu_MMS_2005, Yu_Du_mms_2005,
Hyon_Du_mms_2008}, a set of high-order moments are proposed to 
capture the peak regime of the $\rho(\mb r, t)$, yielding good predictions for the two-dimensional dumbbell system.
However, it is not straightforward to generalize such approximations for complex systems due to the lack
of general relationship between these moments and the stress tensor $\bm\tau_{\rm p}$. On the other hand, the
conformation tensors constructed in the present DeePN$^2$ \emph{are not} the standard moments for the approximation of $\rho(\mb r, t)$;
they are directly learnt from the micro-scale samples that best capture the dynamics of $\bm\tau_{\rm p}$, rather
than recover the high-dimensional $\rho(\mb r, t)$. As a numerical example, we 
employ DeePN$^2$ to a three-bead suspension with intramolecular interactions governed by 
both the bond and angle potentials, see Appendix. Generalization of the learning framework for complex polymer 
fluids will be conducted and presented in the following studies.

\section{Summary}
In this study, we presented a machine learning-based approach  for constructing 
hydrodynamic models for polymer fluids, DeePN$^2$, directly from the micro-scale descriptions. 
While this is only the first step in a long program, the results we obtained have already demonstrated the
potential of such an approach for achieving accuracy and efficiency at the same time.
The construction is based on an underlying micro-scale model.
It respects the symmetries of the underlying physical system.
It is end-to-end, and requires little ad hoc human intervention.
Contrary to conventional wisdom on machine learning models,
the model obtained here is quite interpretable, and in fact shows quite some physical insight.
It has already demonstrated much better accuracy than existing hydrodynamic models in several tests.

Different from the common ML-based approaches for learning the reduced dynamics of complex systems, the 
present approach does not require the time-series samples and provides a generalized form of the objective
tensor derivative with \emph{clear} micro-scale interpretation. This enables us to avoid
the heuristic choices on the objective tensor derivative and the ``black-box'' representations by the 
numerical evaluation of the time-derivatives. These unique features are well-suited for the multi-scale fluid
systems where accurate time-series samples from the micro-scale simulations are often limited.  
While we focused on polymer solutions, the new form of the objective tensor derivative and the 
present learning framework are quite general and can be adapted to other systems of complex fluids and soft matter. 

It should also be noted that what
we discussed is only a first step towards constructing accurate and robust hydrodynamic models for non-Newtonian fluids.
Admittedly, dumbbell suspensions are polymer models with simplified
intramolecular potential and viscoelasticity, applications 
to more realistic micro-scale models will be carried out in  future work.
Among the other issues that remain to be addressed, let us mention coupling the training process with the 
adaptive selection of the training data as was done in MD \cite{Zhang_Lin_PRM_2019},
the automatic choice of the model complexity (e.g. the choice of $n$), the improvement of the underlying
micro-scale model \cite{Lei_Li_PNAS_2016}, and the enhancement of the micro-scale sampling efficiency.  
While some of these will take time, there is no doubt that machine learning, used in the right way,
can help us to tackle the long-standing problem of developing truly reliable hydrodynamic models for complex fluids.

\clearpage
\appendix

\section{Rotational symmetry of the model ansatz and the DNN representation}
In this section, we show that both the modeling ansatz 
and the DNN 
representation 
of the DeePN$^2$ model satisfy the rotational invariance condition. 

\subsection{Rotational invariance from the continuum and microscopic perspective}
Let us consider a symmetric tensor $\mb c \in \mathbb{R}^{3\times 3}$ in two different coordinate 
frames. Frame $1$ is a static inertial frame. We let $\tilde{\mb x}$, 
$\tilde{\mb v} := \tilde{\mb v}(\tilde{\mb x}, t)$, $\tilde{\mb c} := \tilde{\mb c}(\tilde{\mb x}, t)$, 
the position, velocity and $\mb c$ in frame $1$. Framework $2$ is a rotated
frame which is related to frame $1$ by a unitary matrix $\mb{Q}(t)$. We denote $\mb x$, 
$\mb v(\mb x, t)$, $\mb c(\mb x, t)$ the position, velocity and $\mb c$ in 
frame $2$. Accordingly, $\mb x$, $\mb c$ and $\mb v$ follows the transformation rule
\begin{equation}
\begin{split}
&\tilde{\mb x} = \mb Q \mb x \\
&\tilde{\mb v} = \mb Q \mb v(\mb x, t) + \dot{\mb Q}\mb Q^T \tilde{\mb x} \\
&\tilde{\mb c} = \mb Q \mb c(\mb x, t)\mb Q^T.
\end{split}
\label{eq:rotation_transform}
\end{equation}
To construct the dynamics of $\mb c$, we need to choose an objective derivative
$\displaystyle \frac{\mathcal{D} \mb c}{\mathcal{D} t}$ which retains proper rotational symmetry, i.e.,
\begin{equation}
\frac{\mathcal{D} \tilde{\mb c} }{\mathcal{D} t} \big\vert_{\rm frame 1} = 
\mb Q(t) \frac{\mathcal{D} \mb c }{\mathcal{D} t} \big\vert_{\rm frame 2} \mb Q(t)^T.
\label{eq:Der_rot_symm_1}
\end{equation}
For introduction, we choose $\mathcal{D}/\mathcal{D} t$ to be the material derivative of the vector form, i.e., 
$\displaystyle \frac{\rm d}{\rm d t} := \frac{\partial}{\partial t} + \mb v \cdot \nabla$. Accordingly,
Eq. \eqref{eq:Der_rot_symm_1} cannot be satisfied, since 
\begin{equation}
\begin{split}
\frac{\mathcal{D} \tilde{\mb c} }{\mathcal{D} t} \big\vert_{\rm frame 1} 
&= \dot{\mb Q} \mb c \mb Q^T + \mb Q \mb c \dot{\mb Q}^T + \mb Q \frac{\diff \mb c}{\diff t}
\big\vert_{\rm frame~1}  \mb Q^T  \\
&= \dot{\mb Q} \mb c \mb Q^T + \mb Q \mb c \dot{\mb Q}^T + \mb Q \frac{\mathcal{D} \mb c}{\mathcal{D} t}
\big\vert_{\rm frame~2}
\mb Q^T.
\end{split}
\label{eq:Der_rot_symm_2}
\end{equation}
Compared with Eq. \eqref{eq:Der_rot_symm_1}, two additional terms appear in the last equation.
The second identity follows from
\begin{equation}
\begin{split}
\frac{\rm d \mb c}{\rm d t} \big\vert_{\rm frame~1} &= \frac{\partial \mb c(\mb Q^T \tilde{\mb x}, t)}{\partial t} 
+ \tilde{\mb v}\cdot \nabla_{\tilde{\mb x}} \mb c(\mb Q^T \tilde{\mb x}, t)\\
&= \frac{\partial \mb c(\mb x, t)}{\partial t} 
    + \left(\dot{\mb Q}^T \tilde{\mb x} \cdot \nabla_{\mb x}\right) \mb c 
+ \left(\mb Q \mb v(\mb x, t)  + 
        \dot{\mb Q}\mb x\right)\cdot\nabla_{\tilde{\mb x}} \mb c \\
&= \frac{\partial \mb c(\mb x, t)}{\partial t} + \mb v(\mb x, t)\cdot \nabla_{\mb x} \mb c \\
&= \frac{\rm d \mb c}{\rm d t} \big\vert_{\rm frame~2}.
\end{split}
\nonumber
\end{equation}

Alternatively, if we choose $\mathcal{D}/\mathcal{D} t$ to be the objective tensor derivatives 
coupled with $\bm\kappa := \left(\nabla \mb v\right)^T$, e.g., the upper-convected
$\overset{\triangledown}{\mb c} = \frac{\rm d \mb c}{\rm d t} -  \bm\kappa \mb c - \mb c \bm\kappa^T$,  
covariant derivative  
$\overset{\vartriangle}{\mb c} = \frac{\rm d \mb c}{\rm d t} +  \bm\kappa^T \mb c + \mb c \bm\kappa$,
the Jaumann derivative 
$\overset{\circ}{\mb c} = \frac{1}{2}(\overset{\triangledown}{\mb c} + \overset{\vartriangle}{\mb c})$, 
Eq. \eqref{eq:Der_rot_symm_1} is satisfied. For example,  
\begin{equation}
\begin{split}
 \overset{\triangledown}{\mb c} \big\vert_{\rm frame~1} &=  
\dot{\mb Q} \mb c \mb Q^T + \mb Q \mb c \dot{\mb Q}^T + \mb Q 
\frac{\rm d \mb c}{\rm d t} \big\vert_{\rm frame~2} \mb Q^T  \\
& - (\mb Q\bm\kappa\mb Q^T + \dot{\mb Q}\mb Q^T)\mb Q\mb c \mb Q^T
- \mb Q\mb c \mb Q^T (\mb Q\bm\kappa^T\mb Q^T +\mb Q \dot{\mb Q}^T) \\
&= \mb Q \frac{\rm d \mb c}{\rm d t} \big\vert_{\rm frame~2} \mb Q^T 
- \mb Q \bm\kappa \mb c \mb Q^T - \mb Q^T \mb c \bm\kappa^T \mb Q \\
&= \mb Q   \overset{\triangledown}{\mb c} \big\vert_{\rm frame~2} \mb Q^T. 
\end{split}
\nonumber
\end{equation}

On the other hand,  this analysis does not provide us concrete guidance to construct
$\frac{\mathcal{D} \mb c}{\mathcal{D} t}$, since multiple choices such as $\overset{\triangledown}{\mb c}$,
$\overset{\vartriangle}{\mb c}$ and $\overset{\circ}{\mb c}$ all satisfy Eq. \eqref{eq:Der_rot_symm_1}.
To address this issue, we look for a micro-scale perspective based on  the Fokker-Planck equation
to understand the rotational invariance and construct $\frac{\mathcal{D} \mb c}{\mathcal{D} t}$.

Let us consider the Fokker-Planck equation
of a dumb-bell polymer with end-end vector $\mb r$ coupled with flow field $\mb v$. 
By ignoring the external field, the evolution of the density $\rho(\mb r, t)$ is governed by 
\begin{equation}
\frac{\partial \rho(\mb r, t)}{\partial t} = -\nabla\cdot
\left[(\bm\kappa\cdot\mb r)\rho - \frac{2k_BT}{\gamma}\nabla\rho 
- \frac{2}{\gamma}\nabla V_{\rm b} \rho\right],
\label{eq:FK_dumbbell_append} 
\end{equation}
where $\gamma$ is the friction coefficient of the solvent, $V_{\rm b}(r)$ is the intra-molecule
potential energy. 
\begin{proposition}
Eq. \eqref{eq:FK_dumbbell_append} retains rotational invariance under  the transformation by Eq. \eqref{eq:rotation_transform},
i.e.
\begin{equation}
\label{FK-invariance}
\displaystyle \tilde{\rho} := \rho\left(\tilde{\mb r}, t\right)
\big\vert_{\rm frame ~1} \equiv \rho\left(\mb r, t\right)\big\vert_{\rm frame ~2}.
\nonumber
\end{equation}

\end{proposition}

\begin{proof}
\begin{equation}
\begin{split}
&\frac{\partial \tilde{\rho}}{\partial t} + \nabla_{\tilde{\mb r}}\cdot
\left(\left(\tilde{\bm\kappa}\cdot\tilde{\mb r}\right)\tilde{\rho}\right) \big\vert_{\rm frame ~1} \\
&= \frac{\partial \rho}{\partial t} + \dot{\mb Q}^T \tilde{\mb r}\cdot\nabla_{\mb r}\rho
+ \nabla_{\tilde{\mb r}}\cdot
\left(\left(\mb Q \bm\kappa \mb Q^T + \dot{\mb Q}\mb{Q}^T\right)\cdot \mb Q \rho \right)  \\
&=  \frac{\partial \rho}{\partial t} + \dot{\mb Q}^T \mb Q\mb r\cdot\nabla_{\mb r} \rho
+ \nabla_{\mb r}\cdot\left(\bm\kappa\cdot\mb r \rho\right) 
+ \nabla_{\mb r}\cdot\left(\mb Q^T \dot{\mb Q} \mb r \rho\right) \\
&\equiv \frac{\partial \rho}{\partial t} + \nabla_{\mb r}\cdot\left(\bm\kappa\cdot\mb r \rho\right)  \big\vert_{\rm frame ~2}.
\end{split}
\nonumber
\end{equation}
where we have used the fact that $\dot{\mb Q}^T \mb Q$ is anti-symmetric. In addition, it is straightforward to show that
the terms $\nabla^2 \rho$ and $\nabla\cdot \nabla V_{\rm b}(r)$ are invariant. Therefore we have \eqref{FK-invariance}.
\end{proof}

Accordingly, if we define $\mb c$ to be the mean value of a second-order tensor $\mb B(\mb r):\mathbb{R}^3\to \mathbb{R}^{3\times 3}$, the 
dynamics follows
\begin{equation}
\frac{\rm d}{\rm dt} \left\langle \mb B(\mb r)\right\rangle = \bm\kappa:\left\langle
\mb r\nabla_{\mb r} \otimes \mb B\right\rangle + 
\frac{2k_BT}{\gamma}\left\langle \nabla^2 \mb B\right\rangle 
+ 
\frac{2}{\gamma} \left\langle \nabla V_{\rm b}\cdot \nabla\mb B\right\rangle.
\label{eq:FK_B_evoluation_append}
\end{equation}
\begin{proposition}
If $\mb B(r)$ obeys rotational symmetry $\tilde{\mb B} := \mb B(\mb Q^T r) = \mb Q \mb B\mb Q^T$, 
then so does  \eqref{eq:FK_B_evoluation_append}.
\label{prop:B_dumbbell}
\end{proposition}

\begin{proof}
Using Eq. \eqref{eq:Der_rot_symm_2}, the individual terms in frame $1$ follow
\begin{equation}
\frac{\rm d}{\rm dt} \left\langle \tilde{\mb B}\right\rangle \big\vert_{\rm frame~1} = 
\dot{\mb Q} \left\langle \mb B \right\rangle \mb Q^T
+ \mb Q \left\langle \mb B \right\rangle \dot{\mb Q}^T 
+ \mb Q \frac{\rm d}{\rm dt} \left\langle \mb B \right\rangle \big\vert_{\rm frame~2}
\dot{\mb Q}^T.
\label{eq:FK_B_1}
\end{equation}
Note that 
\begin{equation}
\begin{split}
&\tilde{\bm\kappa}:\left\langle \tilde{\mb r}\nabla_{\tilde{\mb r}} \otimes
\tilde{\mb B}\right\rangle \big\vert_{\rm frame 1} \\
& = \left[ \left(\mb Q \bm\kappa \mb Q^T + \dot{\mb Q}\mb{Q}^T\right)
\cdot \mb Q \mb r \right] \cdot \mb Q \nabla_{\mb r} \left(\mb Q \mb B\mb Q^T\right) \\
& = \left(\bm\kappa\cdot\mb r\right)\cdot  \nabla_{\mb r} \left(\mb Q \mb B\mb Q^T\right) +
(\mb Q^T \dot{\mb Q} \mb r)\cdot \nabla_{\mb r} \left(\mb Q \mb B\mb Q^T\right) \\
& = \mb Q (\bm\kappa\cdot\mb r) \cdot \nabla_{\mb r} \mb B\mb Q^T + \mb Q\left(\mb Q^T\dot{\mb Q} \mb B 
+ \mb B \dot{\mb Q}^T \mb {Q}\right)\mb Q^T,
\end{split}
\label{eq:FK_B_2}
\end{equation}
where we have used the relation
\begin{equation}
\left(\mb A \mb r\right) \cdot \nabla \mb B = \mb A \mb B + \mb B \mb A^T,
\end{equation}
if $\mb B$ is a rotational symmetric tensor and $\mb A = \mb Q^T\dot{\mb Q}$ is an anti-symmetric tensor. 

By Eq. \eqref{eq:FK_B_1} and \eqref{eq:FK_B_2}, we see that 
\begin{equation}
\frac{\rm d}{\rm dt} \left\langle \tilde{\mb B}\right\rangle \big\vert_{\rm frame~1}
- \tilde{\bm\kappa}:\left\langle \tilde{\mb r}\nabla_{\tilde{\mb r}}
\otimes
\tilde{\mb B}\right\rangle \big\vert_{\rm frame~1} 
\equiv 
\mb Q \left[
\frac{\rm d}{\rm dt} \left\langle \mb B \right\rangle \big\vert_{\rm frame~2}
- \bm\kappa:\left\langle \mb r \nabla_{\mb r} \otimes
\mb B \right\rangle \big\vert_{\rm frame~ 2}\right]\mb Q^T .
\nonumber
\end{equation}
The rotational symmetry of the other terms follows similarly.
\end{proof}

The above analysis shows that, from the perspective of the Fokker-Planck equation, the
evolution dynamics retains the rotational symmetry. In particular, the term  
$\frac{\rm d}{\rm dt} \left\langle \tilde{\mb B}\right\rangle 
- \tilde{\bm\kappa}:\left\langle \tilde{\mb r}\nabla_{\tilde{\mb r}}
\otimes \tilde{\mb B}\right\rangle $ provides a microscopic perspective
for understanding the objective tensor derivative $\frac{\mathcal{D} \mb B}{\mathcal{D} t}$, which
we use to construct the DNN representation of the constitutive models.

\subsection{DNN representation}
Next we establish a micro-macro correspondence via a set of encoder $\left\{g_i(r)\right\}_{i=1}^n$ 
(see Proposition \ref{thm:g_f_encoder} for details) and, 
accordingly, a set of micro-scale tensor $\mb B_i$ and $\mb c_i$, i.e., 
\begin{equation}
\mb B_i(\mb r) = \left(g_i(r)\mb r\right)\left(g_i(r)\mb r\right)^T,
\quad \mb c_i = \left\langle \mb B_i \right\rangle.
\nonumber
\end{equation}
We will use $\left\{\mb c_i\right\}_{i=1}^n$ to construct the evolution dynamics 
\eqref{eq:FK_B_evoluation_append} via some proper DNN structure which retains the rotational invariance. 
In particular, we consider the fourth-order tensor $\left\langle
\mb r\nabla_{\mb r} \otimes \mb B\right\rangle$ and show that the 
following DNN representation (see also Eq. (12) in main text)
ensures the rotational symmetry of $\frac{\mathcal{D} \mb B}{\mathcal{D} t}$.
For simplicity, the subscript $i$ is ignored and we use $\mb c$ to denote the 
set of conformation tensor $\left\{\mb c_i\right\}_{i=1}^n$.
\begin{proposition}
The following ansatz of $\left\langle\mb r\nabla_{\mb r} \otimes 
\mb B\right\rangle$ 
ensures that the dynamic of evolution of $\mb c$ retains rotational invariance.
\begin{equation}
\boxed{
\left\langle\mb r\nabla_{\mb r}\otimes \mb B\right\rangle 
= \left\langle g(r)^2 \mb r \nabla_{\mb r} 
\otimes
{\color{black}{\mb r}}  {\color{black}{\mb r^T}} \right\rangle 
+ \sum_{i=1}^6 \mb E_1^{(i)} (\mb c) \otimes \mb E_2^{(i)} (\mb c)
}
\nonumber
\end{equation}
where 
$\mb E_1$ and $\mb E_2$ satisfy 
\begin{equation}
\tilde{\mb E}_1 := \mb E_1(\tilde{\mb c}) = \mb Q \mb E_1 \mb Q^T \quad
\tilde{\mb E}_2 := \mb E_2(\tilde{\mb c}) = \mb Q \mb E_2 \mb Q^T.
\nonumber
\end{equation}
\end{proposition}

\begin{proof}
Without loss of generality, we represent the fourth order tensor by the following two bases
\begin{equation}
\begin{split}
&\mb F_1(\mb c)\otimes \mb F_2(\mb c)\otimes \mb F_3(\mb c)  
+ \mb F_1(\mb c)\otimes \left(\mb F_2(\mb c)\otimes \mb F_3(\mb c)\right)^{T_{\{2,3\}}} 
\\
&\mb E_1(\mb c) \otimes \mb E_2(\mb c) \\
&\mb F_1, \mb F_3 \in \mathbb{R}^3, \mb F_2 \in \mathbb{R}^{3\times 3}, \mb E_1, \mb E_2 \in \mathbb{R}^{3\times 3},   
\end{split}
\nonumber
\end{equation}
where the super-script $T_{\{2,3\}}$ represent the transpose between the 2nd and 3rd indices; also
$\mb F_1$, $\mb F_2$, $\mb F_3$, $\mb E_1$ and $\mb E_2$ satisfy the symmetry 
conditions
\begin{equation}
\begin{split}
&\mb F_1(\tilde{\mb c}) = \mb Q \mb F_1 \quad    
\mb F_3(\tilde{\mb c}) = \mb Q \mb F_3 \\
&\mb E_1(\tilde{\mb c}) = \mb Q \mb E_1 \mb Q^T \quad
\mb E_2(\tilde{\mb c}) = \mb Q \mb E_2 \mb Q^T \quad 
\mb F_2(\tilde{\mb c}) = \mb Q \mb F_2 \mb Q^T.
\end{split}
\nonumber
\end{equation}

For the term $\mb E_1(\mb c) \otimes \mb E_2(\mb c)$, we have 
\begin{equation}
\bm\kappa:\mb E_1(\mb c)\otimes \mb E_2(\mb c)\big\vert_{\rm frame 2}
= \Tr(\bm\kappa \mb E_1) \mb E_2
\nonumber
\end{equation}
and
\begin{equation}
\begin{split}
&\tilde{\bm\kappa}:\tilde{\mb E}_1(\mb c)\otimes \tilde{\mb E}_2(\mb c)
\big\vert_{\rm frame 1}\\
&= \left(\mb Q\bm\kappa \mb Q^T + \dot{\mb Q}\mb Q^T \right):\left(\mb Q\mb E_1\mb Q^T 
    \otimes \tilde{\mb E}_2\right)\\
&= \Tr(\bm\kappa \mb E_1)\tilde{\mb E}_2 + \Tr(\dot{\mb Q}\mb Q^T \mb Q\mb E_1\mb Q^T) 
\tilde{\mb E}_2\\
&= \Tr(\bm\kappa \mb E_1)\tilde{\mb E}_2 \\
&\equiv 
\mb Q\left(
\bm\kappa:\mb E_1(\mb c)\otimes \mb E_2(\mb c)\big\vert_{\rm frame 2}\right)\mb Q^T,
\end{split}
\label{eq:G_1_2_rot_symmetry}
\end{equation}
where we have used $\Tr(\dot{\mb Q}\mb Q^T) \equiv 0$. 

For the term $\mb F_1(\mb c)\otimes \mb F_2(\mb c)\otimes \mb F_3(\mb c)  
+ \mb F_1(\mb c)\otimes \left(\mb F_2(\mb c)\otimes \mb F_3(\mb c)\right)^{T_{\{2,3\}}}$,   
we have
\begin{equation}
\bm\kappa:\mb F_1(\mb c)\otimes \mb F_2(\mb c)\otimes \mb F_3(\mb c) \big\vert_{\rm frame 2}
  = \mb F_2 ^T \bm\kappa \mb F_1 \mb F_3^T
\nonumber
\end{equation}
and 
\begin{equation}
\tilde{\bm\kappa}:\tilde{\mb F}_1(\mb c)\otimes \tilde{\mb F}_2(\mb c)\otimes \tilde{\mb F}_3(\mb c)
\big\vert_{\rm frame 1} 
= \mb Q\mb F_2^T \bm\kappa \mb F_1 \mb F_3^T \mb Q^T + \mb Q\mb F_2^T \mb Q^T \dot{\mb Q}
\mb F_1\mb F_3^T \mb Q^T.
\nonumber
\end{equation}
On the other hand, note that 
\begin{equation}
\frac{\diff \tilde{\mb B} }{\diff t} \big\vert_{\rm frame 1} = 
\dot{\mb Q} \mb B \mb Q^T + \mb Q \mb B \dot{\mb Q}^T + \mb Q \frac{\diff \mb B}{\diff t}
\big\vert_{\rm frame~2} \mb Q^T.
\end{equation}
To ensure the rotational symmetry of $\frac{\mathcal{D}\mb B}{\mathcal{D} t}$, we have 
\begin{equation}
\mb F_2 \equiv \mb I, \sum_{i=1}^{K_1} \mb F_1^{(i)} \otimes \mb I \otimes \mb F_3^{(i)} = \left\langle 
g(r)\mb r \otimes \mb I \otimes g(r)\mb r \right\rangle.
\label{eq:choice_F_1_2_3}  
\end{equation}
Hence, we have
\begin{equation}
\begin{split}
&\frac{\rm d}{\rm d t}\tilde{\mb c} - \tilde{\bm\kappa}:
\left( \sum_{i=1}^{K_1} \tilde{\mb F}_1^{(i)}\otimes \tilde{\mb F}_2^{(i)} \otimes \tilde{\mb F}_3^{(i)}
+ \tilde{\mb F}_1^{(i)}\otimes \left(\tilde{\mb F}_2^{(i)} \otimes 
 \tilde{\mb F}_3^{(i)}\right)^{T_{\{2,3\}}}\right)
\big\vert_{\rm frame 1}\\
&\equiv  
\mb Q \left[
\frac{\rm d}{\rm d t}\mb c -  \bm\kappa:
\left( \sum_{i=1}^{K_1} \mb F_1^{(i)}\otimes \mb F_2^{(i)} \otimes \mb F_3^{(i)}
+ \mb F_1^{(i)}\otimes \left(\mb F_2^{(i)} \otimes \mb F_3^{(i)}\right)^{T_{\{2,3\}}}\right)
\big\vert_{\rm frame 2}
\right] \mb Q^T.
\label{eq:F_1_2_3_rot_symmetry}
\end{split}
\end{equation}
Furthermore, using Eq. \eqref{eq:choice_F_1_2_3}, we obtain
\begin{equation}
\sum_{i=1}^{K_1} \mb F_1^{(i)}\otimes \mb F_2^{(i)} \otimes \mb F_3^{(i)}
+ \mb F_1^{(i)}\otimes \left(\mb F_2^{(i)} \otimes \mb F_3^{(i)}\right)^{T_{\{2,3\}}}
= \left\langle g(r)^2 \mb r\nabla_{\mb r} 
\otimes
{\color{black}{\mb r}}  {\color{black}{\mb r^T}} \right\rangle.
\label{eq:F_1_2_3_sum}  
\end{equation}
Accordingly, the remaining part of $\left\langle \mb r\nabla_{\mb r}\otimes \mb B\right\rangle$ is 
expanded by 
\begin{equation}
\left\langle\mb r\nabla_{\color{black}{\mb r}} g(r)^2 
\otimes {\mb r} \mb r^T \right\rangle
= 
\sum_{i=1}^{K_2} 
\mb E_1^{(i)}(\mb c) \otimes \mb E_2^{(i)}(\mb c).
\label{eq:G_1_2_sum}
\end{equation}
where $K_2 = 6$ due to the tensor index symmetry of $1$ and $2$, as well as $3$ and $4$.

Combining Eq. \eqref{eq:F_1_2_3_rot_symmetry}, \eqref{eq:F_1_2_3_sum} and
\eqref{eq:G_1_2_sum}, we conclude that the decomposition 
\begin{equation}
\boxed{
\left\langle\mb r\nabla_{\mb r}\otimes \mb B\right\rangle 
= \left\langle g(r)^2 \mb r \nabla_{\mb r} 
\otimes
{\color{black}{\mb r}}  {\color{black}{\mb r^T}} \right\rangle 
+ \sum_{i=1}^6 \mb E_1^{(i)} (\mb c) \otimes \mb E_2^{(i)} (\mb c)
}
\nonumber
\end{equation}
ensures the rotational invariance in the dynamic equation of $\mb c$.
\end{proof}

Finally, we show that the encoder $\mb f_i(\mb r)$ takes the form $g_i(\vert \mb r \vert) \mb r$ (see 
also Eq. (10) in the main text).
\begin{proposition}
If $\mb f(\mb r): \mathbb{R}^3 \to \mathbb{R}^3$ satisfies 
\begin{equation}
\mb f(\mb Q \mb r) =  \mb Q \mb f(\mb r) 
\nonumber
\end{equation}
for an arbitrary unitary matrix  $\mb Q \in \mathbb{R}^3$,  then
$\mb f(\mb r)$ must take the form $\mb f(\mb r) = g(r)\mb r$, where $g(r): \mathbb{R} \to \mathbb{R}$ is a scalar function 
and $r = \vert \mb r\vert$.
\label{thm:g_f_encoder}
\end{proposition}

\begin{proof}
Let $\mb e_1$, $\mb e_2$ and $\mb e_3$ the basis vectors of the cartesian coordinate space. In particular, we consider $\mb r = r \mb e_1$ and
denote $\mb f(\mb r)$ by $\left(f_1(\mb r), f_2(\mb r), f_3(\mb r)\right)$. By choosing $\mb Q$ to be of the form
\begin{equation}
\mb Q = \begin{pmatrix}
0 &\cos\theta &\sin\theta \\
1 &0 &0\\
0 &-\sin\theta &\cos\theta
\end{pmatrix},
\nonumber
\end{equation}
we have 
\begin{equation}
\mb f(\mb Q\mb r) = 
\begin{pmatrix}
f_1(r\mb e_2) \\ f_2(r\mb e_2) \\ f_3(r\mb e_2)
\end{pmatrix}  
= \begin{pmatrix} 
f_2(r\mb e_1)\cos\theta + f_3(r\mb e_1)\sin\theta\\
 f_1(r\mb e_1) \\
 -f_2(r\mb e_1)\sin\theta + f_3(r\mb e_1)\cos\theta
 \end{pmatrix}.
 \nonumber
\end{equation}
In particular, by choosing $\theta = 0$ and $\theta = \pi$, respectively, we 
get $f_2(r\mb e_1) = f_3(r\mb e_1) = 0$, i.e., $\mb f(r\mb e_1) = \left(f_1(r\mb e_1), 0, 0\right)$.
\end{proof}

\section{The micro-scale model of the dumbbell suspension}
The polymer solution is modeled by suspensions of dumbbell polymer molecules in explicit solvent. The
bond interaction is modeled by the FENE potential, i.e., 
\begin{equation}
V_{\rm b}(r)
=-\frac{k_s}{2}r^2_{0}
\log\left[ 1-\frac{r^2}{r^2_{0}}\right],
\nonumber
\end{equation}
where $k_s$ is the spring constant and $r = \vert \mb r\vert$ and $\mb r$ is the end-end vector between the two beads of a polymer molecule.  
In addition, pairwise interactions are imposed between all particles (except the intramolecular
pairs bonded by $V_b$) under dissipative particle dynamics \cite{Hoogerbrugge_SMH_1992,Groot_DPD_1997}, i.e.,
\begin{equation}
\begin{split}
&\mb{F}_{ij} = \mathbf{F}_{ij}^C +  \mathbf{F}_{ij}^D +   \mathbf{F}_{ij}^D \\
&\mathbf{F}_{ij}^C = 
\begin{cases}
a (1.0 - r_{ij}/r_c) \mathbf{e}_{ij}, & r_{ij} < r_c \\
    0, &r_{ij} > r_c
\end{cases} 
\\
&\mathbf{F}_{ij}^D = 
\begin{cases}
-\gamma w^{D}(r_{ij})(\mathbf{v}_{ij} \cdot \mathbf{e}_{ij}) \mathbf{e}_{ij}, & r_{ij} < r_c\\
    0, &r_{ij} > r_c
\end{cases}  \\  
&\mathbf{F}_{ij}^R = 
\begin{cases}
\sigma w^{R}(r_{ij}) \xi_{ij} \mathbf{e}_{ij}, & r_{ij} < r_c\\
    0, &r_{ij} > r_c
\end{cases},
\end{split}\nonumber
\end{equation}
where $\mathbf{r}_{ij} = \mathbf{r}_{i} - \mathbf{r}_{j}$, $r_{ij} = |\mathbf{r}_{ij}|$, $\mathbf{e}_{ij} = 
\mathbf{r}_{ij}/r_{ij}$, and $\mathbf{v}_{ij} = \mathbf{v}_{i} - \mathbf{v}_{j}$,
$\xi_{ij}$ are independent identically distributed (i.i.d.) Gaussian random variables with zero mean and unit 
variance. 
$\mathbf{F}^C_{ij}$, $\mathbf{F}^D_{ij}$, $\mathbf{F}^R_{ij}$ are the total conservative, dissipative and random 
forces between particles $i$ and $j$, respectively. $r_c$ is the cut-off radius beyond which all interactions vanish. 
The coefficients $a$, $\gamma$ and $\sigma$ represent the strength of 
the conservative, dissipative and random force, respectively. The last two coefficients are coupled with 
the temperature of the system by the fluctuation-dissipation theorem \cite{Espanol_SMO_1995} as $\sigma^2 = 2 \gamma k_BT$. 
Similar to Ref. \cite{Lei_Cas_2010}, the weight functions $w^{D}(r)$ and $w^{R}(r)$ are defined by
\begin{equation}
\begin{split}
w^{D}(r_{ij}) &= \left[w^R(r_{ij})\right]^2 \\
w^{R}(r_{ij}) &= (1.0 - r_{ij} /r_c)^{k}.
\end{split}
\nonumber
\end{equation}
We refer to Ref. \cite{Lei_Karniadakis_JCP_2017} for the details of the reverse Poiseuille flow simulation 
and the calculation of the shear rate dependent viscosity.
In all the numerical experiments, the number density of the solvent particle $n_s$ is set to be $4.0$ and 
the number density of the polymer molecule $n_p$ is set to be $0.5$. 
Other model parameters are given in Tab. \ref{tab:polymer_model_parameter}.
\begin{table}[htbp]
\centering
\caption{Parameters (in reduced unit) of the micro-scale model of the polymer solution}
\begin{tabular*}
{0.45\textwidth}{c @{\extracolsep{\fill}} cccccc}
\hline\hline
&  $a$ & $\gamma$ & $\sigma$ & $k$ & $r_c$\\
\hline
\text{S-S} & $4.0$ & $5.0$ & $1.58$ &$0.25$ & $1.0$\\ 
\text{S-P} & $0.0$ & $40.0$ & $4.47$ &$0.0$ & $1.0$\\ 
\text{P-P} & $0.04$ & $0.01$ & $0.071$ &$0.5$ & $3.5$\\
\hline
\label{tab:polymer_model_parameter}
\end{tabular*}
\end{table}

The training dataset is collected from micro-scale shear flow simulations of the polymer solution
in a domain $[0, 20]\times[0, 20]\times [0, 20]$, with periodic boundary condition imposed in each direction. 
The Lees-Edwards boundary condition \cite{Lees_BC_1972} is used to impose the shear flow rates $\dot{\gamma}$. 
The simulation is run for a production period of $5\times10^4$ with time step $10^{-3}$. $36000$ samples 
of the polymer configurations are collected with $\dot{\gamma}$ uniformly selected between $[0, 0.06]$. 
$32000$ samples are used for training and the remaining ones are used for testing.

\section{Numerical results of a three-bead suspension}
To demonstrate the present DeePN$^2$ method can be applied to systems with high-dimensional configuration space, we 
consider a suspension of 3-bead polymer molecule with the 
intramolecular potential $V_p(\mb r_1, \mb r_2)$ governed by
\begin{equation}
V_p(\mb r_1, \mb r_2) = V_{\rm b}(\mb r_1) + V_{\rm b}(\mb r_2) + V_{\rm a}(\mb r_1, \mb r_2),
\nonumber
\end{equation}
where $V_{\rm b}$ is the FENE bond potential similar to the dumbbell system, $V_{\rm a}$ is the angle
potential defined by
\begin{equation}
V_{\rm a}(\mb r_1, \mb r_2) = \frac{1}{2} k_{a}(\theta - \theta_0)^2,
\nonumber
\end{equation}
where $\theta = \cos^{-1}(\mb r_1\cdot\mb r_2/\vert\mb r_1\vert\cdot\vert\mb r_2\vert)$ is the angle between the two bonds, 
$k_a = 2k_BT$ and $\theta_0 = 2\pi/3$.

We define the generalized conformation tensors 
\begin{equation}
\mb c_i = \left\langle \mb B_i\right\rangle := \left\langle g_i(\vert \mb r_1\vert, \vert \mb r_2\vert, 
\vert \mb r_{12}\vert)^2 {\mb r'}_i {\mb r''}_i^T\right\rangle, 
\label{eq:c_triplet}
\nonumber
\end{equation}
where ${\mb r'}_i$ and ${\mb r''}_i$ are chosen to be either $\mb r_1$ or $\mb r_2$.  
Similar to the dumbbell model, we set $\mb c_1 = \left\langle \mb r_1 \mb r_1^T\right\rangle$, 
$\mb c_2 = \left\langle \mb r_1 \mb r_2^T\right\rangle$, 
and choose the eigen-space of $\mb c_1$ as the reference frame for the training process. 
We employ the constructed model to  simulate the reverse Poiseuille flow. The setup is similar to the dumbbell suspension. 
Fig. \ref{fig:angle_MD_ML} shows  the evolution of the velocity 
profile and the mean cosine value
(i.e., ${\rm Tr}(\mb c_2)/{\rm Tr}(\mb c_1)$) with the 
body force $f_{\rm ext} = 0.0066$.  Predictions from DeePN$^2$ agree well 
with the MD results. In contrast, predictions from the Hookean model show apparent deviations.

This numerical example shows that the present DeePN$^2$ method is not limited by the high-dimensionality 
of the polymer configuration space, in contrast with the previous approaches based on the direct approximation of the probability
density $\rho(\mb r, t)$.  More sophisticated learning framework applicable to the general multi-bead polymer suspension with 
complex intramolecular potential requires further investigations, and will be presented in following works.

\begin{figure}[htbp]
\centering
\includegraphics[trim=60 20 100 60,clip,scale=0.25]{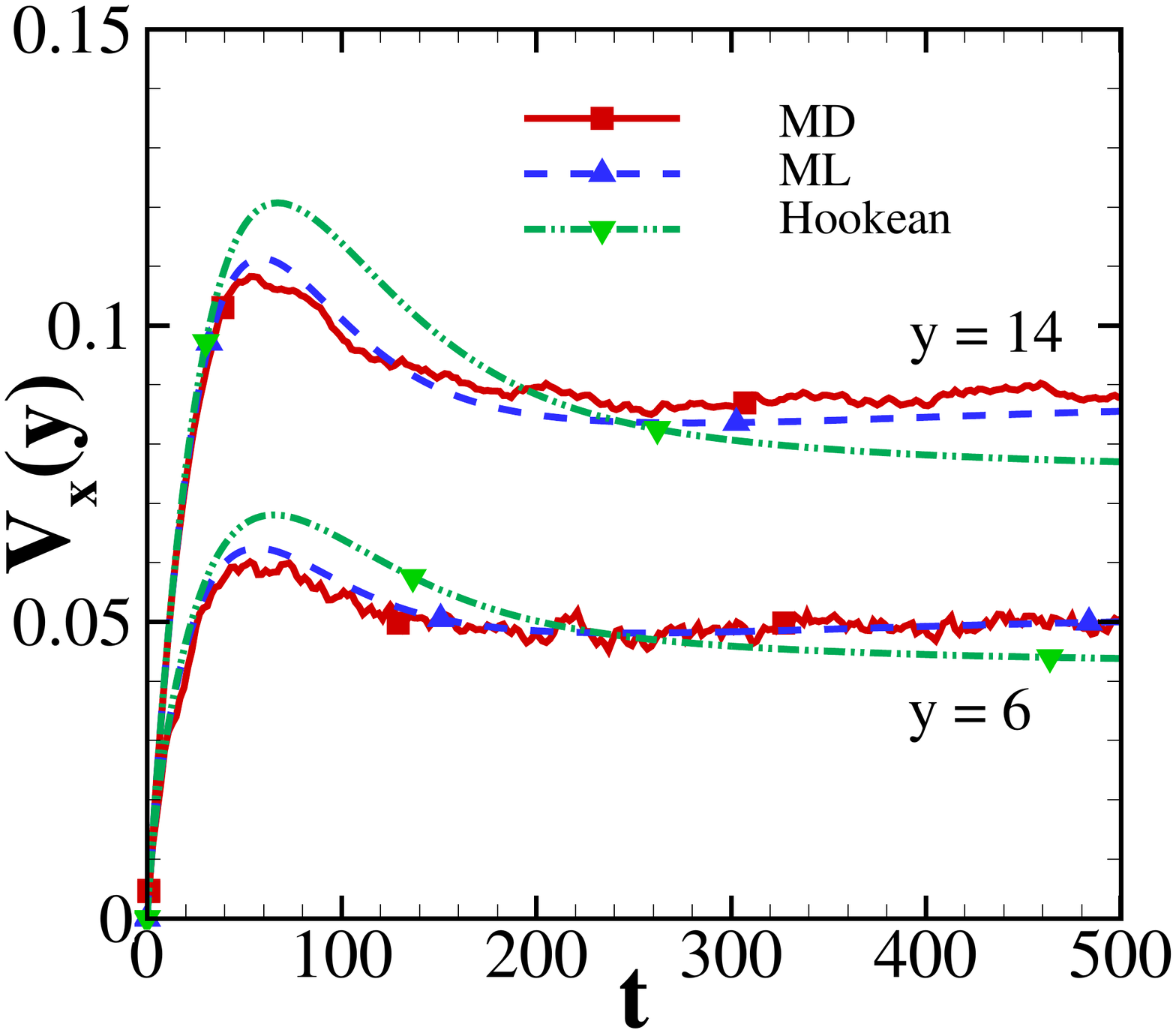}
\includegraphics[trim=60 20 100 60,clip,scale=0.25]{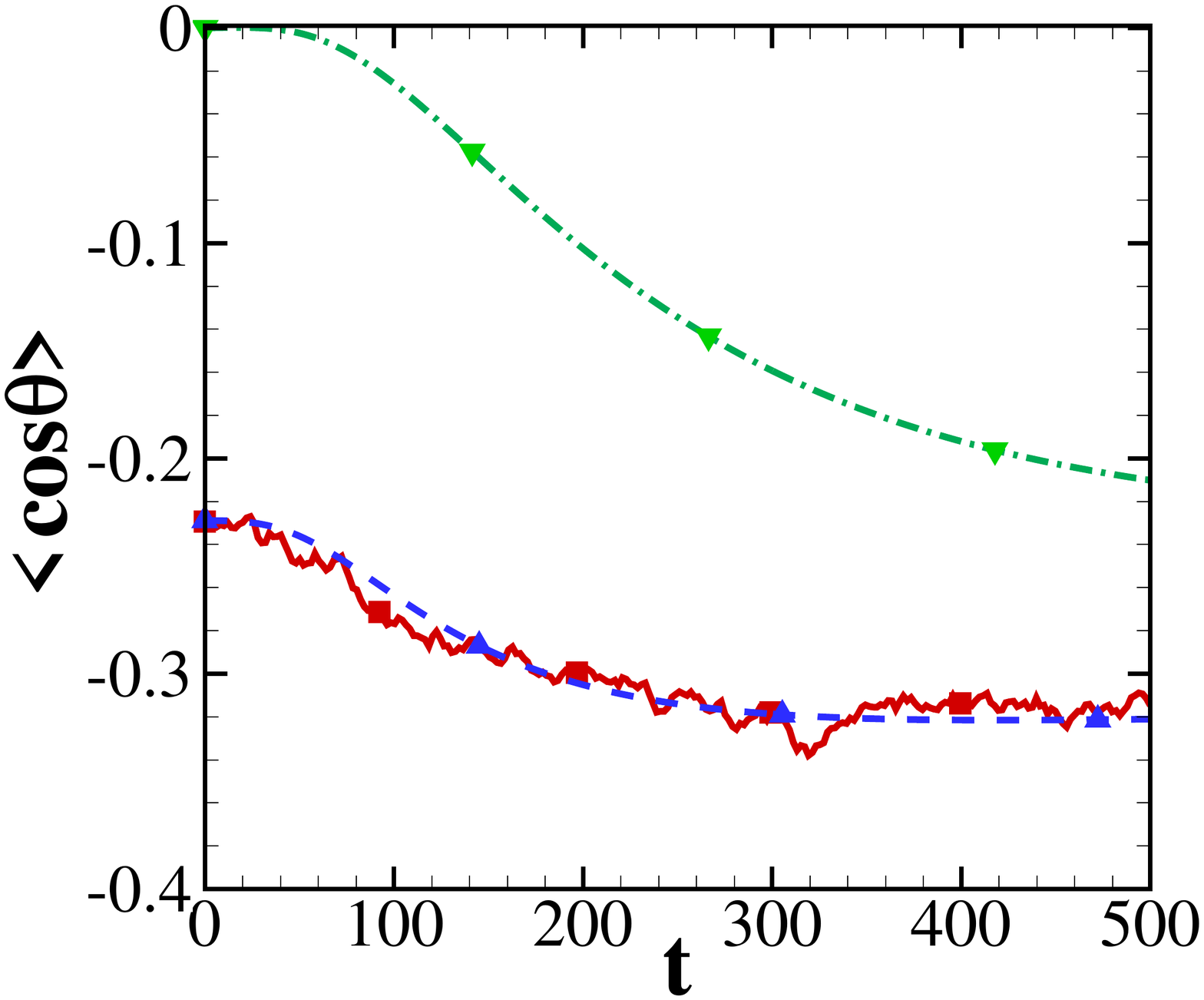}
\caption{Comparison of  the time evolution of reverse Poiseuille flow of a three-bead suspension 
obtained from direction MD simulation, DeePN$^2$ and 
Hookean model. The parameters of the Hookean model are chosen such that the polymer average bond length matches the 
MD results. \textbf{Left}: Velocity at $y = 6$ and $14$. \textbf{Right}: Average value of $\cos\theta$ at $y=6$.}
\label{fig:angle_MD_ML}
\end{figure}

\section{Training procedure}
\label{sec:training}
The constructed DeePN$^2$ model is represented by various DNNs for the encoders $\left\{g_j(r)\right\}_{j=1}^{n}$,
stress model $\mb G$,  evolution dynamics $\left\{\mb H_{1,j}\right\}_{j=1}^n$, $\left\{\mb H_{2,j}\right\}_{j=1}^n$  
and the $4$th order tensors $\left\{\mathcal{E}_j\right\}_{j=1}^{n}$ of the objective tensor derivatives. In particular, by 
choosing $g_1(r) \equiv 1$, $\mb G \propto \mb H_{2,1}$ so we do not
need to train $\mb G$ separately. The loss function is defined by
\begin{equation}
L = \lambda_{H_1} L_{H_1} + \lambda_{H_2} L_{H_2} + \lambda_{\mathcal{E}} L_{\mathcal{E}},
\nonumber
\end{equation}
where $\lambda_{H_1}$, $\lambda_{H_2}$ and $\lambda_{\mathcal{E}}$ are hyperparameters specified later. For each
training batch of $m$ training samples,  $L_{H_1}$, $L_{H_2}$, $L_{\mathcal{E}}$ of the dumbbell system are given by 
\begin{equation}
\begin{split}
L_{H_1} &= \sum_{i=1}^m \sum_{j=1}^{n}\left\Vert \mb V^{(i)} \widehat{\mb H}_{1,j}(\widehat{\mb c}_1^{(i)}, \cdots \widehat{\mb c}_n^{(i)}) {\mb V^{(i)}}^T    
- \left\langle \nabla^2_{\mb r} \mb B_j(\mb r) \right\rangle^{(i)} \right\Vert^2  \\
L_{H_2} &= \sum_{i=1}^m \sum_{j=1}^{n}\left\Vert  \mb V^{(i)} \widehat{\mb H}_{2,j}(\widehat{\mb c}_1^{(i)}, \cdots \widehat{\mb c}_n^{(i)}) {\mb V^{(i)}}^T     
- \left\langle \nabla V_{\rm b}(r) \cdot \nabla_{\mb r} \mb B_i(\mb r) \right\rangle^{(i)} \right\Vert^2  \\
L_{\mathcal{E}} &= \sum_{i=1}^m \sum_{j=1}^{n} \left\Vert \sum_{k}\mb V^{(i)} \widehat{\mb E}_{1,j}^{(k)} (\widehat{\mb c}_1^{(i)}, \cdots \widehat{\mb c}_n^{(i)}) {\mb V^{(i)}}^T   
\otimes  \mb V^{(i)} \widehat{\mb E}_{2,j}^{(k)} (\widehat{\mb c}_1^{(i)}, \cdots \widehat{\mb c}_n^{(i)}) {\mb V^{(i)}}^T \right .\\
&-
\left . \left\langle\mb r\nabla_{\color{black}{\mb r}} g(r)^2 
\otimes {\mb r} \mb r^T \right\rangle^{(i)} 
\right \Vert^2,  
\end{split}
\end{equation}
where $\Vert \cdot \Vert^2$ denotes the total sum of squares of the entries in the tensor. $\mb V^{(i)}$
is the matrix composed of the eigenvectors of $\mb c_1 = \langle \mb r\mb r^T\rangle$ of the $i\mhyphen$th sample. 

Furthermore, we note that $\mb c$, $\mb H_1$, $\mb H_2$, $\mb E_1$ and $\mb E_2$ are all symmetric. Accordingly, the 
DNN inputs are composed
of the upper-triangular parts of the $\mb c$ and the outputs are the upper-triangular parts of the representation tensors.  
Specifically, $\left\{g_j\right\}_{j=1}^n$, $\left\{\mb H_{1,j}\right\}_{j=1}^n$, $\left\{\mb H_{2,j}\right\}_{j=1}^n$,
$\left\{\mb E_{1,j}, \mb E_{2,j}\right\}_{j=1}^n$ are represented by the $8\mhyphen$layer fully-connected DNNs. The 
number of neurons in the hidden layers are set to be {\small $(120, 120, 120, 120, 120, 120)$, $(300, 300, 400, 400, 300, 300)$, 
$(400, 600, 600, 600, 600, 400)$, $(300, 300, 300, 300, 300, 300)$}, respectively. The activation function is taken
to be the hyperbolic tangent. 

The DNNs are trained by the Adam stochastic gradient descent method \cite{Kingma_Ba_Adam_2015} for $400$ epochs, using
$75$ samples per batch size. The initial learning rate is $1.8\times 10^{-4}$ and decay rate is $0.8$ per $9000$ steps.
The hyper-parameters $\lambda_{H_1}$, $\lambda_{H_2}$ and $\lambda_{\mathcal{E}}$ are chosen in the following two ways. In the 
first setup, we set them to be constant throughout the training process, e.g., 
$\lambda_{H_1} = \lambda_{H_2} = \lambda_{\mathcal{E}} = 1/3$. In the second setup, the hyper-parameters are updated 
every $N_{\lambda}$ epochs by
\begin{equation}
\lambda_{H_1} = \frac{\widetilde{L}_{H_1}}{\widetilde{L}_{H_1} + \widetilde{L}_{H_2}+ \widetilde{L}_{\mathcal{E}}},~
\lambda_{H_2} = \frac{\widetilde{L}_{H_2}}{\widetilde{L}_{H_1} + \widetilde{L}_{H_2}+ \widetilde{L}_{\mathcal{E}}},~ 
\lambda_{\mathcal{E}} = \frac{\widetilde{L}_{\mathcal{E}}}{\widetilde{L}_{H_1} + \widetilde{L}_{H_2}+ \widetilde{L}_{\mathcal{E}}},
\end{equation}
where $\widetilde{L}_{(\cdot)}$ denotes the mean of the loss  during the past $N_{\lambda}$ epochs. For the present
study, both approaches achieve a loss $L$ smaller than $1\times10^{-4}$ and the root of relative loss less than 
$1.6\times10^{-2}$. More sophisticated choices of $\lambda_{H_1}$, $\lambda_{H_2}$ and $\lambda_{\mathcal{E}}$ as well as
other formulation of $L$  will be investigated in future work.

\section{Computational cost}
We consider two dynamic processes:  relaxation to quasi-equilibrium and the development of 
the reverse Poiseuille flow. For relaxation to quasi-equilibrium, the micro-scale simulation is conducted in
a domain $[0, 10]\times[0, 10]\times [0, 10]$ (in reduced unit), which is mapped into a volume unit in the continuum  
DeePN$^2$, Hookean and FENE-P models. All simulations are run for a production period of $360$ (in reduced unit). 
For the case of the reverse Poiseuille flow (RPF), the microscale simulation is conducted in
a domain $[0, 40]\times[0, 80]\times [0, 40]$.  The simulations of the continuum  
DeePN$^2$, Hookean, and FENE-P models are conducted by mapping the domain into $20$  
volume units along y direction. All simulations are run for a production period 
of $550$. 
The computational cost for both systems is reported in Tab. \ref{tab:computational_cost_all_models}.
All simulations are performed on Michigan State University HPCC supercomputer with 
Intel(R) Xeon(R) CPU E5-2670 v2. 
\begin{table}[htbp]
\centering
\caption{Computational cost (in CPU-second) using the MD model and the continuum
DeePN$^2$, FENE-P and Hookean models.}
\begin{tabular*}
{0.6\textwidth}{c @{\extracolsep{\fill}} ccccc}
\hline\hline
&  {\rm MD} & DeePN$^2$ & FENE-P & Hookean & \\
\hline
\text{Quasi-equilibrium} & ~$2.35\times 10^4$ & $4.1$ & $0.56$ &$0.51$ \\ 
\text{RPF (dumbbell)} & ~$9.24\times 10^6$ & $85.6$ & $10.2$ &$9.7$ \\ 
\hline
\label{tab:computational_cost_all_models}
\end{tabular*}
\end{table}

We note that the size of the volume unit is chosen empirically in the continuum models of the 
flow systems considered in the present work. Our sensitivity studies show that the numerical results
of the DeePN$^2$ model agree well with the full MD when the average number of polymer within 
a unit volume is greater than $200$. For all the cases, the computational cost of the DeePN$^2$
model is less than $0.05\%$ of the computational cost of the full MD simulations and less than 10 times
the cost of empirical continuum models.

\begin{acknowledgments}
We thank Jiequn Han, Chao Ma, Linfeng Zhang for helpful discussions.
The work of Huan Lei is supported in part by the Extreme Science and Engineering Discovery Environment (XSEDE) Bridges
at the Pittsburgh Supercomputing Center through allocation DMS190030 and the High Performance Computing Center at Michigan State University.  
The work of Weinan E and Lei Wu is supported in part by a gift to Princeton University from
iFlytek.
\end{acknowledgments}


%

\end{document}